\documentclass[letter,12pt]{article}

\usepackage{graphics}
\usepackage{graphicx} 
\usepackage{rotating}
\usepackage{dcolumn}
\usepackage{longtable}
\usepackage{multirow}
\usepackage{amsmath}
\usepackage{amsthm}
\usepackage{amssymb}
\usepackage{xspace}
\usepackage{textcase}
\usepackage[margin=1in]{geometry}
\usepackage{appendix}
\usepackage{float}
\usepackage[latin1]{inputenc}
\usepackage{tikz}
\usetikzlibrary{trees}
\usetikzlibrary{arrows}
\usepackage{subcaption}

\DeclareMathOperator*{\argmin}{arg\,min}
\newtheorem{theorem}{Theorem}[section]
\newtheorem{corollary}{Corollary}[theorem]
\newtheorem{lemma}[theorem]{Lemma}

\begin{document}
\title{\vspace{-36pt}}
\title{Hypergames and Cyber-Physical Security for Control Systems}
\author{Craig Bakker \and
Arnab Bhattacharya \and
Samrat Chatterjee \and
Draguna L. Vrabie}
\maketitle

\begin{abstract}
The identification of the Stuxnet worm in 2010 provided a highly publicized example of a cyber attack used to damage an industrial control system physically.  This raised public awareness about the possibility of similar attacks against other industrial targets -- including critical infrastructure.  In this paper, we use hypergames to analyze how adversarial perturbations, like those used by Stuxnet, can be used to manipulate a system that employs optimal control.  Hypergames form an extension of game theory that enables us to model strategic interactions where the players may have significantly different perceptions of the game(s) they are playing.  Past work with hypergames has been limited to relatively simple interactions consisting of a small set of discrete choices for each player, but here, we apply hypergames to larger systems with continuous variables.  We find that manipulating constraints can be a more effective attacker strategy than directly manipulating objective function parameters.  Moreover, the attacker need not change the underlying system to carry out a successful attack -- it may be sufficient to deceive the defender controlling the system.  It is possible to scale our approach up to even larger systems, but the ability to do so will depend on the characteristics of the system in question, and we identify several characteristics that will make those systems amenable to hypergame analysis.
\end{abstract}

\section{Introduction}
\subsection{Stuxnet and Cyber-Physical Security}

The Stuxnet worm was identified in 2010 as a piece of malware that targeted a very specific Industrial Control System (ICS) -- namely, uranium enrichment infrastructure \cite{nourian18jsr,symantec11tr}.  This may not have been the first cyber attack to cause physical damage to an ICS, but it was highly publicized.  As such, Stuxnet brought the potential physical consequences of cyber attacks into the public eye.

Stuxnet was highly sophisticated.  Part of its sophistication lay in its strategy for obtaining access to its targets: it exploited four 0-day vulnerabilities, compromised two digital certificates, and propagated itself through networks and removable devices \cite{symantec11tr}.  Once it reached a control system, it continued to act stealthily.  Stuxnet fed fake data to the ICS to disguise malicious actions \cite{symantec11tr,howser14cp} and limited its attacks to avoid detection \cite{karnouskos11cp}.  The goal of Stuxnet was not to cause catastrophic failure, which would have been easier.  Rather, it exploited the physical vulnerabilities as well as the cyber vulnerabilities inherent in the ICS.

Stuxnet forced analysts to consider the risk associated with these kinds of cyber attacks.  If we understand risk as the product of consequence, vulnerability, and threat, we can address each of those components separately.  The potential for significant consequence is clear: many industrial processes, including critical infrastructure systems (e.g., the power grid), rely on Supervisory Control and Data Acquisition (SCADA) software and ICSs.  These systems are also vulnerable.  Updates can be risky because they may cause previously functional systems to produce new errors \cite{karnouskos11cp}, and even if this is not the case, taking the system in question offline to perform the updates may be difficult or infeasible \cite{nourian18jsr}.  There is a tradeoff between security and ease of use, and a knowledge gap between cyber security specialists and control engineers can compound this.

There are two more factors that increase the vulnerability of ICSs to cyber attack.  Firstly, industrial systems are often serviced by outside contractors, and the devices (computers, USB drives, etc.) used by those contractors can provide a malware vector that bypasses traditional cyber security measures such as air gaps \cite{symantec11tr}.  Secondly, industry standardization also reduces uncertainty for potential attackers; complexity, heterogeneity, and uncertainty make it more difficult for attackers to design successful attacks.

Most of the uncertainty regarding the risk of cyber attacks on ICSs has to do with threat.  The old consensus was that these systems were too specialized to attack \cite{karnouskos11cp}.  Stuxnet, for example, required a great deal of specialized knowledge about the control systems in question \cite{symantec11tr}.  In the case of terrorism, for example, it is easier to build a bomb than to write code that will cause comparable physical destruction.  However, Stuxnet showed that these kinds of attacks are possible for those determined to carry them out.

\subsection{Hypergames}

Game theory is a branch of mathematics that looks at strategic interactions between rational entities.  It has seen considerable use in economic \cite{roth02jsr} and security \cite{sandler2003terrorism} applications.  A fundamental premise of strategic games in game theory is that all of the players are seeing and playing the same game.  This is not always true, though.  Belief manipulation plays a key role in some strategic interactions.  In other cases, not all player objectives may be common knowledge.  This necessitates understanding more completely players' perceptions of the game(s) they are playing; one way to model this is through hypergames \cite{bennett80jsr}.  

Hypergames allow players to play different games and can account for differences in their perceptions of the same game without considering uncertainty probabilistically.  For example, one group of players may distinguish between certain actions while another group considers those actions all to be identical.  On the other hand, some players may not be aware of the existence of other players in the game (or may not be aware of all of those other players' actions).  Hypergames essentially enable us to extend the concept of rationality to a bounded information situation.  This, in turn, makes it possible for a given player to exploit another player's misperceptions.  In analyzing the (potentially) different games that each player is playing, though, we are still able to apply game theoretic concepts and thus build on existing game theory research.  We can describe a two-player game as 

\begin{gather}
G_{A,B} = \left(\mathcal{P},\mathcal{S},\mathcal{U}\right)
\end{gather}

\begin{gather}
\mathcal{P} = \left\{A,B\right\} \\
\mathcal{S} = \left\{S_A,S_B\right\} \\
\mathcal{U} = \left\{u_A,u_B\right\}
\end{gather}

\noindent where $A$ and $B$ are the players, $S_A$ and $S_B$ are those players' respective action spaces, and $u_A,u_B: S_A \times S_B \rightarrow \Re$ are their respective payoff functions, which provide a partial ordering over $S_A \times S_B$ for each player.  We can describe a first level hypergame as

\begin{gather}
H_{A,B} \left(A,B,G_{A,B}\right) = \left\{ p\left(A,G_{A,B}\right),p\left(B,G_{A,B}\right)\right\}
\end{gather}

\noindent where $p\left(A,G_{A,B}\right)$ is $A$'s perception of $G_{A,B}$.  The condition $p\left(A,G_{A,B}\right) \neq p\left(B,G_{A,B}\right)$ could be caused by discrepancies such as $p \left(A,\left\{A,B\right\}\right) = \left\{A\right\}$, which would indicate that $A$ is not aware of $B$'s presence.  We can also describe perceptions about perceptions.  For example, $p\left(AB,u_A\right)$ is $A$'s perception of $B$'s perception of $A$'s utility function.  For a first level hypergame, there are misperceptions, but the players are not aware of those misperceptions:

\begin{gather}
p\left(A,G_{A,B}\right) \neq p\left(B,G_{A,B}\right) \\
p\left(AB,G_{A,B}\right) = p\left(A,G_{A,B}\right) \\
p\left(BA,G_{A,B}\right) = p\left(B,G_{A,B}\right)
\end{gather}

For a second level hypergame, at least one player is aware of the misperceptions.  For example, if $A$ is aware of the misperceptions but $B$ is not, we have

\begin{gather}
p\left(AB,G_{A,B}\right) \neq p\left(A,G_{A,B}\right) \\
p\left(BA,G_{A,B}\right) = p\left(B,G_{A,B}\right)
\end{gather}

Player $B$ then plays $p\left(B,G_{A,B}\right)$ while $A$ plays the hypergame

\begin{gather}
H_{A,AB} \left(A,AB,G_{A,B}\right) = \left\{p\left(A,G_{A,B}\right),p\left(AB,G_{A,B}\right)\right\}
\end{gather}

The overall solution to a hypergame can then be calculated by correctly aggregating the equilibrium solutions to the players' perceived (hyper)games.  In a first level hypergame, for example, the equilibrium solution is $\left(x_A,x_B\right)$, where $x_A$ is $A$'s equilibrium strategy for the game $p\left(A,G_{A,B}\right)$ and $x_B$ is $B$'s equilibrium strategy for the game $p\left(B,G_{A,B}\right)$.  For the second level hypergame described above, $x_A$ would be $A$'s optimal strategy for $H_{A,AB}\left(A,AB,G_{A,B}\right)$, while $x_B$ would still be $B$'s equilibrium strategy for $p\left(B,G_{A,B}\right)$.  These concepts extend naturally to higher level hypergames and additional players.  See Kovach et al. \cite{kovach15jsr} and Gutierrez et al. \cite{gutierrez15jsr} for more details.

Approaches such as reflexive control \cite{novikov14bk}, Mirage Equilibria \cite{sakovics01jsr}, and $k$-level reasoning \cite{camerer04jsr,stahl95jsr} have been applied to systems that may not have common knowledge (and thereby incorporate a kind of bounded rationality).  Despite some differences in notation and nomenclature, these approaches all incorporate hierarchies of beliefs (e.g., Player 1's beliefs about Player 2's beliefs).  However, the first two, along with hypergames, differ somewhat from $k$-level reasoning with respect to the accuracy of the player perceptions.  In $k$-level reasoning, the focus is on the degree to which one player anticipates another.  In principle, this approach does not rule out the possibility that a given player might misperceive the nature of the game (payoff structure, available actions, etc.), but in practice, this is not a key consideration.  For hypergames, this \emph{is} a key consideration.  The concept of a subjective game (i.e., $p\left(A,G_{A,B}\right)$) is central to hypergame analysis, and belief hierarchies exist to support that; the same is true for reflexive control and Mirage Equilibria.

For example, a key hypergame result is that hypergame equilibrium solutions can be stable under misperceptions \cite{sasaki08cp}.  In these cases, each player does what the other players expect -- which can happen even when the players' perceptions differ or are erroneous -- and thus there is no motivation for players to update their perceptions.  This is similar to a conjectural equilibrium \cite{sakovics01jsr} in that players do not know what they do not know.  In a repeated hypergame context, then, these equilibria are stable, and extending belief hierarchies to higher and higher levels would not necessarily change that.  Using the formalism we employed previously, a hypergame equilibrium is stable if $p\left(A,x_B\right) = x_B$ and $p\left(B,x_A\right) = x_A$,  which need not imply that $p\left(A,G_{A,B}\right) = p\left(B,G_{A,B}\right)$.

Hypergames have been used to study water resource management \cite{okada85jsr}, supply chain relationships \cite{graham92jsr}, and cyber attacks \cite{house10jsr}.  Some research has also looked at connecting hypergames with other branches of game theory.  Kanazawa et al. \cite{kanazawa07jsr} studied an evolutionary version of hypergames.  This included calculating evolutionarily stable strategies and defining hypergame replicator dynamics.  Sasaki and Kijima \cite{sasaki12jsr,sasaki16jsr} showed how hypergames can be reformulated as Bayesian games (at least in some cases).  In doing so, though, they identified reasons why it may be advantageous to avoid that reformulation.  Firstly, hypergames can provide a simpler and more natural epistemic representation of the game's players; the treatment of unawareness, for example, can be more convincing than in the Bayesian case.  Secondly, there are some hypergame solution concepts, such as stability under misperception, that do not map to the Bayesian reformulation.  The topic of misperception has also led to research into how repeated hypergames can be used to improve or update perceptions \cite{sasaki08cp}.  House and Cybenko used both hidden Markov models and a maximum entropy approach \cite{house10jsr}.  Takahashi et al., on the other hand, used a genetic algorithm \cite{takahashi99cp}.  Generally speaking, though, the hypergame literature is relatively small; Kovach et al. provided a review of the field \cite{kovach15jsr}.  Moreover, all of the examples that we have seen have involved hypergames with a relatively small number of discrete choices.  Solving for the equilibrium solutions, then, has involved hand calculations and/or exhaustive enumeration.

\subsection{Aim and Motivation}

The goal of this paper is to show how hypergames can be used in optimal control where the control system in question is subject to adversarial perturbations and to demonstrate how this analysis can apply to Stuxnet-like attacks.  This research contributes to ongoing work in optimal control by showing how manipulating controller perceptions can function as an attacker strategy; the attacker actually uses the control system against itself.  These analyses then highlight weaknesses in the control system -- weaknesses that are vulnerable to attack even if they might not be vulnerable to random events.  This research also advances hypergame research in two ways.  Firstly, it brings hypergames to bear on a new application area (i.e., optimal control) -- one rather different than the examples in previous papers.  Secondly, it applies hypergame concepts to systems of significantly greater complexity than previous hypergame research has used.  The examples in this paper have continuous variables, and the second example is a discrete-time optimal control problem with time-varying variables.  Both problems, moreover, require using numerical optimization methods to find hypergame equilibria.  Taking hypergames to this level of complexity makes the hypergame concept more viable as a tool for analyzing real systems and not just toy problems.

This kind of investigation is highly relevant to addressing Stuxnet-like attacks from a control perspective.  Leaving aside IT-based cyber security concerns, let us assume that an attacker has access to at least part of an ICS.  Can we then characterize the kind of damage that that attacker could produce, and can we design control procedures that minimize that damage?  In this paper, we focus primarily on the former but touch upon the latter; we intend to address the latter more fully in later work.  ICSs provide examples of (potentially high-impact) cyber-physical systems where control provides the connection between the `cyber' and 'physical' components.  The idea behind this research, then, is not to replace traditional cyber security methods but rather to recognize that control systems can be used to provide another layer of robustness to attack if those control systems are designed to do so and that the physical weaknesses accessible through cyber means can be analyzed by looking at the control model.

\section{Static Problem Formulation}

To demonstrate some of the concepts of this paper, we consider a static optimization problem constrained within an operating envelope, which is represented as an inequality constraint: 

\begin{gather}
\min \limits_{u} J\left(u,\theta\right) \\ 
g(u,c)\leq 0 
\label{operating constraint}
\end{gather}

\noindent where $u$ is the vector of decision variables, $\theta$ is the vector of objective function parameters, and $c$ is the vector of operating envelope parameters.  Note that $g$ may be a vector of constraint equations $g_l$, $l = 1,2,\ldots$, in which case (\ref{operating constraint}) is equivalent to $g_l\left(u,c\right) \leq 0 \ \forall \ l$.

\subsection{Objective Function Manipulation}

Here, we will consider a situation where the attacker can manipulate the defender's observation of objective function parameters; $\hat{\theta} = \theta + \Delta \theta$, where the vector $\hat{\theta}$ denotes the quantities that the defender observes.  The attacker optimization is then

\begin{gather}
\max \limits_{\Delta \theta}  J\left(\hat{u}^*,\theta\right) 
\label{att start} \\
\frac{1}{2} \left\|\Delta \theta\right\|^2 \leq \delta_{\theta,max} 
\label{att theta cons} \\
\hat{u}^* = \argmin_{\hat{u}} \left( J\left(\hat{u},\hat{\theta}\right) : g(\hat{u},c) \leq 0 \right)
\label{def opt}
\end{gather}

\noindent where (\ref{def opt}) describe what the attacker expects the defender's optimization to be and (\ref{att theta cons}) is a constraint on the attacker's manipulations, which is a reasonable assumption in a context of limited attack budgets or when attack detection mechanisms are present in the system.  This constitutes a second level hypergame.  If $A$ represents the attacker and $D$ represents the defender, we have 

\begin{gather}
p\left(D,\theta\right) = \hat{\theta} \neq \theta = p\left(A,\theta\right) \\
p\left(D,\left\{A,D\right\}\right) = \left\{D\right\} = p\left(AD,\left\{A,D\right\}\right)
\end{gather}

If the defender knows of the attacker, this leads to a higher level hypergame, where 

\begin{equation}
p\left(DAD,\left\{A,D\right\}\right) = p\left(AD,\left\{A,D\right\}\right) = \left\{D\right\}
\end{equation}

The defender's optimization is

\begin{gather}
\min \limits_{u} J\left(u,\hat{\theta} - \Delta \theta\right) \\
g(u,c) \leq 0 
\end{gather}

\noindent The true $\theta$ values are unknown to the defender, but the defender calculates the $\Delta \theta$ values by solving what is believed to be the attacker's problem: (\ref{att start})-(\ref{def opt}).

\begin{gather}
\max \limits_{\Delta \theta} J\left(\hat{u}^*,\theta\right) \\
\frac{1}{2}\left\|\Delta \theta\right\|^2 \leq \delta_{\theta,max}\\
\hat{u}^* = \argmin_{\hat{u}} \left( J\left(\hat{u},\hat{\theta}\right) : g(\hat{u},c) \leq 0 \right)
\end{gather}

Given that the defender only knows $\hat{\theta}$, not $\theta$, solving the attacker's problem to determine $\Delta \theta$ will require using $\theta = \hat{\theta} - \Delta \theta$.  As a further extension, we consider the scenario where the attacker manipulates the defender's perceptions of $\theta$, the defender knows that the attacker is doing this, and the attacker knows that the defender is anticipating the attacker's perturbations.  We refer to this as a `double-bluff' manipulation here and in the rest of the paper.  This problem leads us to a multi-level optimization problem:

\begin{gather}
\max_{\Delta \theta} J\left(u^*,\theta\right) \\
\frac{1}{2} \left\| \Delta \theta\right\|^2 \leq \delta_{\theta,max} \\
u^* = \argmin_u \left(J\left(u,\tilde{\theta}\right) : g(u,c) \leq 0 \right)
\end{gather}

\noindent subject to

\begin{gather}
\max_{\Delta \hat{\theta}} J\left(\hat{u}^*,\tilde{\theta}\right) \\
\frac{1}{2}\left\|\Delta \hat{\theta}\right\|^2 \leq \delta_{\theta,max} \\
\hat{u}^* = \argmin_{\hat{u}} \left(J\left(\hat{u},\tilde{\theta} + \Delta \hat{\theta}\right): g(\hat{u},c) \leq 0\right)
\end{gather}

\noindent where $p\left(D,\theta\right) = \tilde{\theta} = \hat{\theta} - \Delta \hat{\theta}$ is the defender's estimate of the true value of $\theta$.  There are many other potential combinations of misperceptions that we could also consider.  Note that the defender's perceived cost (i.e., objective function value) may differ from the true cost in some cases.

For the purpose of comparison, we can model the attacker manipulating the true value of $\theta$:

\begin{gather}
\max_{\Delta \theta} \min_u J\left(u,\theta + \Delta \theta\right) \\
\left\| \Delta \theta\right\|^2 \leq \delta_{\theta,max} \\
g\left(u,c\right)\leq 0
\end{gather}

In this case, there are no misperceptions, and the situation is simply a zero-sum game, not a hypergame.

\subsection{Constraint Manipulation}
\label{Constraint Manipulation}

The previous section involved the attacker manipulating parameters in the objective function.  In this case, there is a significant difference between manipulating the true values and the defender's perceptions.  If the attacker is manipulating the constraints, however, then the distinction changes.  If the attacker alters the constraint to be more restrictive, then it does not matter whether the manipulation is of the real constraint or of the defender's perceptions -- both actions lead to the same result (assuming that the defender abides by the constraint), and the perceived cost is the true cost in both cases.  If the attacker alters the constraint to be less restrictive, the results are less clear.  If the attacker manipulates the defender perception, the control process may hit a physical limit and/or damage the system trying to reach an infeasible state.  This could be modelled by having some kind of large penalty function for violations of the true constraint.  Manipulating the true constraint in such a way as to relax it may be impossible if the constraint is a physical limitation of the system.  For this section, we specify that the attacker can manipulate the defender's perception of parameters in the constraint ($\hat{c} = c + \Delta c$ are the quantities that the defender perceives).  These perturbations are then subject to a cost function as with the objective function parameter perturbations.

\begin{gather}
\frac{1}{2} \left\| \Delta c\right\|^2 \leq \delta_{c,max}
\end{gather}

\subsubsection{Maximizing Cost}

If the attacker is manipulating the defender's perceptions to maximize defender cost, this results in a series of multi-level optimization problems, corresponding to second or higher level hypergames, analogous to those described in the previous section.  If the attacker is deceiving an unsuspecting defender, we have

\begin{gather}
\max \limits_{\Delta c}  J\left(u^*,\theta\right) \\
\frac{1}{2} \left\|\Delta c\right\|^2 \leq \delta_{c,max}\\
u^* =  \argmin_u \left( J\left(u,\theta\right) : g\left(u,\hat{c}\right) \leq 0 \right)
\end{gather}

If the defender is aware of the attack, we have

\begin{gather}
\min \limits_{u} J\left(u,\theta\right) \\
g\left(u,\hat{c} - \Delta c\right) \leq 0
\end{gather}

\noindent subject to

\begin{gather}
\max \limits_{\Delta c} J\left(\hat{u}^*,\theta\right) \\
\frac{1}{2}\left\|\Delta c\right\|^2 \leq \delta_{c,max}\\
\hat{u}^* = \argmin_{\hat{u}} \left(J\left(\hat{u},\theta\right): g\left(\hat{u},\hat{c}\right) \leq 0 \right)
\end{gather}

In a situation analogous to that described in the previous section, the defender only knows $\hat{c}$, not $c$, so solving the attacker's problem to determine $\Delta c$ will require using $c = \hat{c} - \Delta c$.  If the attacker is aware that the defender is anticipating an attack, the resulting problem is

\begin{gather}
\max_{\Delta c} J\left(u,\theta\right) \\
\frac{1}{2} \left\| \Delta c\right\|^2 \leq \delta_{c,max} \\
u^* = \argmin_u \left(J\left(u,\theta\right) : g(u,\tilde{c}) \leq 0 \right)
\end{gather}

\noindent subject to

\begin{gather}
\max_{\Delta \hat{c}} J\left(\hat{u}^*,\theta\right) \\
\frac{1}{2}\left\|\Delta \hat{c}\right\|^2 \leq \delta_{c,max} \\
\hat{u}^* = \argmin_{\hat{u}} \left(J\left(\hat{u},\theta \right): g(\hat{u},\tilde{c} + \Delta \hat{c}) \leq 0\right)
\end{gather}

\noindent where $p\left(D,c\right) = \tilde{c} = \hat{c} - \Delta c$ is the defender's estimate of the true value of $c$.

\subsubsection{Breaking the System}
The attacker could also try to cause the defender to deviate maximally from the operating envelope constraint in the interest of causing a catastrophic failure.  We refer to this as attempting to break the system.  If the attacker is deceiving an unsuspecting defender, we have

\begin{gather}
\max \limits_{\Delta c}  \gamma^T g\left(u^*,c\right) \\
\frac{1}{2} \left\|\Delta c\right\|^2 \leq \delta_{c,max}\\
u^* = \argmin_u \left(J\left(u,\theta\right) : g\left(u,\hat{c}\right) \leq 0 \right)
\end{gather}

\noindent where $\gamma^T g\left(u,c\right)$ indicates a weighted sum for a vector-valued $g$.  If the defender is aware of the attack, we have

\begin{gather}
\min \limits_{u} J\left(u,\theta\right) \\
g\left(u,\hat{c} - \Delta c\right) \leq 0
\end{gather}

\noindent subject to

\begin{gather}
\max \limits_{\Delta c} \gamma^T g\left(\hat{u}^*,c\right) \\
\frac{1}{2}\left\|\Delta c\right\|^2 \leq \delta_{c,max}\\
\hat{u}^* = \argmin_{\hat{u}} \left( J\left(\hat{u},\theta\right) : g\left(\hat{u},\hat{c}\right) \leq 0 \right)
\end{gather}

If the attacker is aware that the defender is anticipating an attack, the resulting problem is

\begin{gather}
\max_{\Delta c}  \gamma^T g\left(u^*,c\right) \\
\frac{1}{2} \left\| \Delta c\right\|^2 \leq \delta_{c,max} \\
u^* = \argmin_u \left(J\left(u,\theta\right) : g(u,\tilde{c}) \leq 0 \right)
\end{gather}

\noindent subject to

\begin{gather}
\max_{\Delta \hat{c}}  \gamma^T g\left(\hat{u}^*,c\right) \\
\frac{1}{2}\left\|\Delta \hat{c}\right\|^2 \leq \delta_{c,max} \\
\hat{u}^* = \argmin_{\hat{u}} \left(J\left(\hat{u},\theta \right): g(\hat{u},\tilde{c} + \Delta \hat{c}) \leq 0\right)
\end{gather}

\noindent where $\tilde{\theta}$ is defined as before.  There are various other possibilities in the same vein involving asymmetric information or false beliefs.

\subsection{Analytical Results}

\subsubsection{Objective Function Perturbations}

In this section, we will show that the defender can be robust with respect to manipulated perceptions of $\theta$.  Let us assume that $g\left(u,c\right)$ is convex for $c\geq0$ and 

\begin{gather}
J\left(u,\theta\right) = \sum_k \theta_k f_k\left(u\right)
\end{gather}

\noindent where each $f_k\left(u\right)$ is convex.  The optimization is therefore convex for $\theta \geq 0$, and the optimality conditions

\begin{gather}
\sum_k \frac{\partial f_k}{\partial u} \theta_k + \sum_l \lambda_l \frac{\partial g_l}{\partial u} = \theta^T \frac{\partial f}{\partial u} + \lambda^T \frac{\partial g}{\partial u} = 0 \\
0 \leq \lambda_l \perp g_l\left(u,c\right) \leq 0 \ \forall \ l
\end{gather}

\noindent are both necessary and sufficient; $\lambda$ is the vector of Kuhn-Tucker multipliers.  Let us also define

\begin{gather}
R\left(u\right) = \left\{ \left. \frac{\partial f_k}{\partial u}\right|_u : k = 1,2,\ldots, n_{\theta}\right\} \\
S\left(u\right) = \left\{ l : \lambda_l > 0\right\} \\
S'\left(u\right) = \left\{ l : g_l\left(u,c\right) = 0\right\} \\
T\left(u\right) = \left\{\left.\frac{\partial g_l}{\partial u}\right|_u : l \in S\left(u\right) \right\} \\
T'\left(u\right) = \left\{\left.\frac{\partial g_l}{\partial u}\right|_u : l \in S'\left(u\right) \right\}
\end{gather}

\noindent where $R$ and $T$ are sets of vectors, $S$ is a set of indices denoting the positive $\lambda_l$ values at $u$, and $S'$ is a set of indices denoting the active set at $u$.  Note that $S\left(u\right) \subseteq S'\left(u\right)$, and $S\left(u\right) \neq S'\left(u\right)$ only if there are active constraints with corresponding multipliers that are zero.

\begin{lemma}
\label{lem 1}
Assume that $u^* \in \argmin \limits_u \left(J\left(u,\theta\right): g\left(u,c\right) \leq 0 \right)$ and that $\hat{\theta} = \theta + \Delta \theta \geq 0$.  If there exists $\Delta \lambda \geq -\lambda$ such that

\begin{gather}
\Delta \theta^T \left.\frac{\partial f}{\partial u}\right|_{u^*} + \Delta \lambda^T \left.\frac{\partial g}{\partial u}\right|_{u^*} = 0 \\
\Delta \lambda_l g_l \left(u^*,c\right) = 0 \ \forall \ l
\end{gather}

\noindent then $u^* \in \argmin \limits_u \left(J\left(u,\hat{\theta}\right): g\left(u,c\right) \leq 0 \right)$ and $\hat{\lambda} = \lambda + \Delta \lambda$ are the new Kuhn-Tucker multipliers.
\end{lemma}

\begin{proof}

If 

\begin{gather}
\theta^T \left.\frac{\partial f}{\partial u}\right|_{u^*} + \lambda^T \left.\frac{\partial g}{\partial u}\right|_{u^*} = 0 \\
\Delta \theta^T  \left.\frac{\partial f}{\partial u}\right|_{u^*} + \Delta \lambda^T \left.\frac{\partial g}{\partial u}\right|_{u^*} = 0
\end{gather}

\noindent then for $\hat{\theta} = \theta + \Delta \theta$

\begin{gather}
\left(\hat{\theta}^T - \Delta \theta^T \right) \left.\frac{\partial f}{\partial u}\right|_{u^*} + \left(\lambda^T - \Delta \lambda^T + \Delta \lambda^T \right) \left.\frac{\partial g}{\partial u}\right|_{u^*} = 0 \\
\hat{\theta}^T \left.\frac{\partial f}{\partial u}\right|_{u^*} + \left(\lambda^T + \Delta \lambda^T \right) \left.\frac{\partial g}{\partial u}\right|_{u^*} = 0
\end{gather}

Furthermore, since $\Delta \lambda \geq - \lambda$ and $\Delta \lambda_l g_l \left(u^*,c\right) = 0 \ \forall \ l$,

\begin{gather}
\hat{\theta}^T \left.\frac{\partial f}{\partial u}\right|_{u^*} + \hat{\lambda}^T \left.\frac{\partial g}{\partial u}\right|_{u^*} = 0 
\label{opt lem 1} \\
0 \leq \hat{\lambda}_l \perp g_l\left(u^*,c\right) \leq 0 \ \forall \ l
\label{opt lem 2}
\end{gather}

\noindent where $\hat{\lambda} = \lambda + \Delta \lambda$.  Since $\hat{\theta}\geq 0$ and $J\left(u,\hat{\theta}\right)$ and $g\left(u,c\right)$ are convex, the optimization

\begin{gather}
\min_u J\left(u,\hat{\theta}\right) \\
g\left(u,c\right) \leq 0
\end{gather}

\noindent is convex, and the optimality conditions (\ref{opt lem 1})-(\ref{opt lem 2}) are necessary and sufficient.  $u^*$ satisfies these conditions, so $u^* \in \argmin \limits_u \left(J\left(u,\hat{\theta}\right): g\left(u,c\right) \leq 0 \right)$.
\end{proof}

\begin{lemma}
\label{cor 2}
If $\text{span} \left(R\left(u^*\right)\right) \subseteq \text{span}\left(T\left(u^*\right)\right)$, then there exists $r>0$ such that for $\left\| \Delta \theta \right\|_p \leq r$, $p>0$, $u^* \in \argmin \limits_u \left(J\left(u,\theta\right): g\left(u,c\right) \leq 0 \right)$ implies that $u^* \in \argmin \limits_u \left(J\left(u,\hat{\theta}\right): g\left(u,c\right) \leq 0 \right)$, where $\hat{\theta} = \theta + \Delta \theta$.
\end{lemma}

\begin{proof}
Let us define the matrix $A$ such that the rows of $A$ are the vectors $\frac{\partial g_l}{\partial u} \in T\left(u^*\right)$.  If $\text{span} \left(R\left(u^*\right)\right) \subseteq \text{span}\left(T\left(u^*\right)\right)$, then any linear combination of $\frac{\partial f_k}{\partial u} \in R\left(u^*\right)$ exists within $\text{span}\left(T\left(u^*\right)\right)$, which is the rowspace of $A$.  This implies that for any $\Delta \theta$, there exists $\Delta \lambda$ such that

\begin{gather}
\sum_k \Delta \theta_k \left.\frac{\partial f_k}{\partial u}\right|_{u^*} + \sum \limits_{l \in S\left(u^*\right)} \Delta \lambda_l \left. \frac{\partial g_l}{\partial u} \right|_{u^*} = \Delta \theta^T \left. \frac{\partial f}{\partial u} \right|_{u^*} + b^T A = 0 \\
\Delta \lambda_l = 0, \ l \notin S \left(u^*\right)
\end{gather}

\noindent and if $A^+$ is the Moore-Penrose pseudo-inverse of $A$, then

\begin{gather}
b^T = - \Delta \theta^T \left. \frac{\partial f}{\partial u} \right|_{u^*} A^+
\end{gather}

\noindent satisfies this exactly because $\Delta \theta^T \left. \frac{\partial f}{\partial u} \right|_{u^*}$ is in the rowspace of $A$.  Define

\begin{gather}
\lambda_{min} = \min \limits_{l \in S\left(u^*\right)} \lambda_l
\end{gather}

By definition, $\lambda_{min} > 0$.  If $\left\| \Delta \lambda \right\|_p \leq \lambda_{min}$, then

\begin{gather}
\max_l \left|\Delta \lambda_l\right| = \left\| \Delta \lambda\right\|_{\infty} \leq \left\| \Delta \lambda \right\|_p \leq \lambda_{min} , \ p >0
\end{gather}

Therefore, $\left\| \Delta \lambda \right\|_p \leq \lambda_{min}$ implies that $\max \limits_l \left|\Delta \lambda_l\right| \leq \lambda_{min} $ and thus $\Delta \lambda_l \geq - \lambda_{min} \geq - \lambda_l \ \forall \ l$.  If

\begin{gather}
\left\| \Delta \theta\right\|_p \leq \frac{\lambda_{min}}{\left\| \frac{\partial f}{\partial u} A^+\right\|_p} = r
\end{gather}

\noindent then

\begin{gather}
\left\| \Delta \lambda \right\|_p = \left\| \Delta \theta^T \frac{\partial f}{\partial u} A^+ \right\|_p \leq \left\| \Delta \theta \right\|_p \left\| \frac{\partial f}{\partial u} A^+ \right\| \leq \lambda_{min}
\end{gather}

By Lemma \ref{lem 1}, $u^* \in \argmin \limits_u \left(J\left(u,\hat{\theta}\right): g\left(u,c\right) \leq 0 \right)$.

\end{proof}

\begin{corollary}
\label{cor 1}
Define the matrix $A$ such that the rows of $A$ are the vectors $\frac{\partial g_l}{\partial u} \in T\left(u^*\right)$.  If $A$ is invertible, then there exists $r>0$ such that for $\left\| \Delta \theta \right\|_p \leq r$, $p>0$, $u^* \in \argmin \limits_u \left(J\left(u,\theta\right): g\left(u,c\right) \leq 0 \right)$ implies that $u^* \in \argmin \limits_u \left(J\left(u,\hat{\theta}\right): g\left(u,c\right) \leq 0 \right)$, where $\hat{\theta} = \theta + \Delta \theta$.
\end{corollary}

\begin{proof}
If $A$ is invertible, then the rows of $A$ are linearly independent and $\text{span} \left(T \left(u^*\right)\right) = R^{n_u}$, where $u \in R^{n_u}$, and thus $\text{span} \left(R\left(u^*\right)\right) \subseteq \text{span}\left(T\left(u^*\right)\right)$.  This satisfies the conditions of Lemma \ref{cor 2}, and thus the same conclusions follow.

\end{proof}

\begin{lemma}
\label{lem 2}
The set $\Theta \left(u^*\right) = \left\{ \hat{\theta} : u^* \in \argmin \limits_u \left(J\left(u,\hat{\theta}\right): g\left(u,c\right) \leq 0 \right), \hat{\theta} \geq 0\right\}$ is unbounded and convex if it is non-empty.
\end{lemma}

\begin{proof}
$J\left(u,\theta\right)$ is linear in $\theta$, so $J\left(u,c\hat{\theta}\right) = c J\left(u,\hat{\theta}\right)$ for any positive scalar $c$.  Optimal solutions are invariant with respect to scalar multiples of the objective function:

\begin{gather}
\argmin_u \left(J\left(u,\hat{\theta}\right): g\left(u,c\right) \leq 0 \right) = \argmin_u \left(c J\left(u,\hat{\theta}\right): g\left(u,c\right) \leq 0 \right) \nonumber \\
= \argmin_u \left(J\left(u,c\hat{\theta}\right): g\left(u,c\right) \leq 0 \right)
\end{gather}

Therefore, for any $\hat{\theta} \in \Theta \left(u^*\right)$ and any positive scalar $c$, $c \hat{\theta} \in \Theta \left(u^*\right)$.  Thus, $\Theta \left(u^*\right)$ is unbounded if it is non-empty.  Furthermore, for fixed $u^*$, the optimality conditions

\begin{gather}
\hat{\theta}^T \left.\frac{\partial f}{\partial u}\right|_{u^*} + \hat{\lambda}^T \left.\frac{\partial g}{\partial u}\right|_{u^*} = 0 \\
\hat{\lambda}_l = 0 \ l \notin S'\left(u^*\right) \\
\hat{\lambda}_l \geq 0 \ l \in S'\left(u^*\right)
\end{gather}

\noindent form a set of linear inequalities in $\hat{\lambda}$ and $\hat{\theta}$; because $u^*$ is fixed, we can disregard $g\left(u^*,c\right) \geq 0$.  The space of $\hat{\lambda}$ and $\hat{\theta}$ that satisfy these constraints is therefore convex.  Since this space is convex, for any $\left(\hat{\theta}_1,\hat{\lambda}_1\right)$ and $\left(\hat{\theta}_2,\hat{\lambda}_2\right)$ in this space

\begin{gather}
\left(\alpha\hat{\theta}_1 + \left(1-\alpha\right) \hat{\theta}_2,\alpha\hat{\lambda}_1 + \left(1-\alpha\right) \hat{\lambda}_2\right) , \ \alpha \in \left[0,1\right]
\end{gather}

\noindent remains in $\Theta \left(u^*\right)$.  Thus for any $\hat{\theta}_1, \hat{\theta_2} \in \Theta \left(u^*\right)$, $\left(\alpha\hat{\theta}_1 + \left(1-\alpha\right) \hat{\theta}_2\right) \in \Theta \left(u^*\right)$, so $\Theta \left(u^*\right)$ is convex.
\end{proof}

\begin{theorem}
If $\text{span} \left(R\left(u^*\right)\right) \subseteq \text{span}\left(T\left(u^*\right)\right)$ and $u^* \in \argmin_u \left(J\left(u,\theta\right): g\left(u,c\right) \leq 0 \right)$, there exists a convex, unbounded set of $\Delta \theta$ such that $u^* \in \argmin \limits_u \left(J\left(u,\theta + \Delta \theta \right): g\left(u,c\right) \leq 0 \right)$.
\end{theorem}

\begin{proof}
By Lemma \ref{cor 2}, if $\text{span} \left(R\left(u^*\right)\right) \subseteq \text{span}\left(T\left(u^*\right)\right)$ and $u^* \in \argmin_u \left(J\left(u,\theta\right): g\left(u,c\right) \leq 0 \right)$, then there exists $r>0$ such that for $\left\| \Delta \theta \right\|_p \leq r$, $p>0$, $u^* \in \argmin \limits_u \left(J\left(u,\theta + \Delta \theta\right): g\left(u,c\right) \leq 0 \right)$.  Therefore, the set 

\begin{gather}
\Theta \left(u^*\right) = \left\{ \hat{\theta} : u^* \in \argmin \limits_u \left(J\left(u,\hat{\theta}\right): g\left(u,c\right) \leq 0 \right), \hat{\theta} \geq 0\right\}
\end{gather}

\noindent is non-empty.  By Lemma \ref{lem 2} if $\Theta \left(u^*\right)$ is non-empty, it is unbounded and convex.
\end{proof}

\begin{lemma}
If $\text{span} \left(R\left(u^*\right)\right) \nsubseteq \text{span} \left(T'\left(u^*\right)\right)$ for $u^* \in \argmin \limits_u \left(J\left(u,\theta\right): g\left(u,c\right) \leq 0 \right)$, then for any $\epsilon > 0$, there exists $\Delta \theta$ such that $\left\| \Delta \theta\right\| < \epsilon$ and $u^* \notin \argmin \limits_u \left(J\left(u,\hat{\theta}\right): g\left(u,c\right) \leq 0 \right)$.
\end{lemma}

\begin{proof}
Assume that for sufficiently small $\epsilon>0$, there is no $\Delta \theta$ such that $0 < \left\| \Delta \theta\right\| < \epsilon$ and $u^* \notin \argmin \limits_u \left(J\left(u,\hat{\theta}\right): g\left(u,c\right) \leq 0 \right)$.  Then for sufficiently small $\Delta \theta$, there exists $\Delta \lambda$ such that

\begin{gather}
\Delta \theta^T \left.\frac{\partial f}{\partial u}\right|_{u^*} + \Delta \lambda^T \left.\frac{\partial g}{\partial u}\right|_{u^*} = 0 \\
\Delta \lambda_l \geq - \lambda_l \ l \in S'\left(u^*\right) \\
\Delta \lambda_l = 0 \ l \notin S'\left(u^*\right)
\end{gather}

Since $u^*$ is fixed, the active set cannot change.  Let us define the matrix $A$ such that the rows of $A$ are the vectors $\frac{\partial g_l}{\partial u}\in T'\left(u^*\right)$ and define the vector $b$ such that the elements of $b$ are $\Delta \lambda_l, \ l \in S'\left(u^*\right)$.  Then

\begin{gather}
\Delta \lambda^T \left.\frac{\partial g}{\partial u}\right|_{u^*} = b^T A \\
\Delta \theta^T \left.\frac{\partial f}{\partial u}\right|_{u^*} + b^T A = 0
\end{gather}

If $\text{span} \left(R\left(u^*\right)\right) \nsubseteq \text{span} \left(T\left(u^*\right)\right), \ l \in S'\left(u^*\right)$, then there exists $\Delta \theta_0$ such that $\Delta \theta_0^T \frac{\partial f}{\partial u} \notin \text{span} \left(T\left(u^*\right)\right)$ and therefore

\begin{gather}
\Delta \theta_0^T \left.\frac{\partial f}{\partial u}\right|_{u^*} + b^T A \neq 0 \ \forall \ b
\end{gather}

Moreover, for any such $\Delta \theta_0$, there exists $\Delta \theta = c \Delta \theta_0$ such that for any $c > 0$

\begin{gather}
c \Delta \theta_0^T \left.\frac{\partial f}{\partial u}\right|_{u^*} + b^T A \neq 0 \ \forall \ b
\end{gather}

Since $\left\| c \Delta \theta_0\right\| = c \left\| \Delta \theta_0\right\|$, for any $\epsilon > 0$, there exists $\Delta \theta = \frac{\epsilon}{\left\| \Delta \theta_0\right\|} \Delta \theta_0$ such that

\begin{gather}
\Delta \theta^T \left.\frac{\partial f}{\partial u}\right|_{u^*} + b^T A \neq 0 \ \forall \ b
\end{gather}

Because the optimality conditions are necessary and sufficient, and because these conditions cannot be satisfied, $u^* \notin \argmin \limits_u \left(J\left(u,\theta + \Delta \theta \right): g\left(u,c\right) \leq 0 \right)$, and thus the lemma is proved by contradiction.
\end{proof}

If small $\Delta \theta$ values change the value of $u^*$ but not the active set, it is possible to calculate the $\frac{\partial u}{\partial \Delta \theta}$ for the optimal solution by differentiating the optimality conditions.  This provides us with a linear system that we can solve to calculate $\frac{\partial u}{\partial \Delta \theta}$, and $u^*\left(\Delta \theta\right)$ will be smooth and well-defined as long as the active set does not change.  We can therefore compare this kind of system with one that is impervious to these small changes.  For such a system, the measure of the `safe' range is conservative, but outside of it, continuous changes in $\hat{\theta}$ could result in discrete jumps in $u^*$ as the active set changes.  Furthermore, if $J\left(u,\theta\right)$ is nonlinear in $\theta$ but still convex for all $\theta \geq 0$, then it may possible to produce similar proofs for this case, but this would require further assumptions regarding the dependence of $J$ on $\theta$.

\subsubsection{Constraint Function Manipulations}

Unfortunately, manipulations of $c$ are not subject to the same kinds of robustness that manipulations of $\theta$ are.  This is essentially a consequence of the discussion at the beginning of Section \ref{Constraint Manipulation}: manipulating the defender's perception of the constraints produces the same change in the decision variables as changing the true constraints would as long as the defender abides by the perceived constraints.  For example, 

\begin{gather}
g_l \left(u,c\right) = 0, \ l \in S\left(u\right) \\
\sum_i \frac{\partial g_l}{\partial u_i} \frac{\partial u_i}{\partial c_j} + \frac{\partial g_l}{\partial c_j} = 0
\end{gather}

Therefore, if $\frac{\partial g_l}{\partial c_j} \neq 0$, then $\frac{\partial u_i}{\partial c_j} \neq 0$.

\subsection{Test Problem}

As a demonstration, we consider minimizing power consumption for a fan in an HVAC system.  A problem like this could form a component in a larger HVAC system, possibly as a subsystem subject to repeated optimization under changing parameter values.  The baseline defender optimization problem is

\begin{gather}
\min \limits_{m,p} \theta_1 m + \theta_2 m^2 + \theta_3 p \\
\frac{1}{2} \left[ \left(m-c_m\right)^2 + \left(p-c_p\right)^2 - c_r^2\right] \leq 0
\label{static cons}
\end{gather}

\noindent where $m$ is the mass flow rate, $p$ is the static pressure, the $\theta$ values are power consumption parameters for the fan, and $c_m$, $c_p$, and $c_r$ are parameters defining the operating envelope.  The attacker can introduce perturbations $\Delta \theta_i$ such that $\hat{\theta}_i = \theta_i + \Delta \theta_i$ and $\frac{1}{2} \left\| \Delta \theta\right\|^2_2 \leq \delta_{\theta,max}$
or perturbations $\Delta c_m,\Delta c_p,\Delta c_r$ such that $\hat{c}_m = c_m + \Delta c_m$, $\hat{c}_p = c_p + \Delta c_p$, $\hat{c}_r = c_r - \Delta c_r$, and $\frac{1}{2} \left\| \Delta c \right\|^2_2 \leq \delta_{c,max}$.  Note the negative sign in $\hat{c}_r$.  This deviates slightly from our convention above, but it also helps to simplify later calculations in some ways, and it does not ultimately change the results.  In our computations in the rest of the paper, we use $\theta_1 = \theta_2 = 1$, $\theta_3 = 2$, $c_m = c_p = 5$, and $c_r^2 = 10$.  The $\frac{1}{2}$ constant in (\ref{static cons}) does not change the mathematical properties of the optimization, but it, too, simplifies some of the calculations used later in this paper; see Appendix \ref{Fan Description} for these calculations.

\section{Dynamic Optimization}

\subsection{Model Formulation}

We now bring hypergames to bear on a Model Predictive Control (MPC) problem, where the control objective is to minimize a cost function subject to state dynamics constraints and operational constraints over a time horizon of length $\tau$:

\begin{gather}
\min_{u^t} \sum \limits_{t=1}^{\tau} J(u^t,x^t,\theta) 
\label{dyn generic start}\\
x^t = f(x^{t-1},u^t,\alpha^t,\beta)
\label{dynamics} \\
x^{\tau} - x^0 = 0 \\
g\left(x^t,u^t,\alpha^t,\beta\right)\leq 0 
\label{dyn generic end}
\end{gather}

\noindent where $u^t$ are the control decision variables, $x^t$ are the states of the system, $\alpha^t$ are the system disturbances, and $\beta$ are the model parameters. We assume that $\beta$ and $\alpha^t$ can be affected by adversarial perturbations.  The attacker can either perturb the defender's perception of parameters $\beta$ to maximize cost (`Static Attack') or perturb the defender's perception of $\alpha^t$ to maximize cost ('Dynamic Attack'). The perturbations denoted $\Delta \beta$, and $\Delta \alpha^t$ are bounded by constraints, normalized as appropriate if they have different orders of magnitude; such constraints are then with respect to \textit{relative} perturbations on those parameters.

\begin{gather}
\frac{\Delta \beta}{\beta} \equiv \left[ \frac{\Delta \beta_1}{\beta_1} \ \frac{\Delta \beta_1}{\beta_1} \ \ldots \right]^T \\
\frac{1}{2} \left\|\frac{\Delta \beta}{\beta}\right\|^2 \leq \delta_{\beta,max} \\
\frac{1}{2} \sum \limits_{t=1}^{\tau} \|\Delta \alpha^t\|_2^2 \leq \delta_{\alpha,max}
\end{gather}

The static attack problem is

\begin{gather}
\max_{\Delta \beta} \sum \limits_{t=1}^{\tau} J(u^t,x^t,\theta) \\
\frac{1}{2} \left\|\frac{\Delta \beta}{\beta}\right\|^2 \leq \delta_{\beta,max} \\
x^t = f(x^{t-1},u^t,\alpha^t,\beta) \\
\hat{x}^0 = x^0
\end{gather}

\noindent subject to

\begin{gather}
\min_{u^t} \sum \limits_{t=1}^{\tau} J(u^t,\hat{x}^t,\theta) \\
\hat{x}^t = f(\hat{x}^{t-1},u^t,\alpha^t,\hat{\beta}) \\
\hat{x}^{\tau} - \hat{x}^0 = 0 \\
g\left(\hat{x}^t,u^t,\alpha^t,\hat{\beta}\right)\leq 0 
\end{gather}

This is a second level hypergame where $p\left(D,\beta\right) = \hat{\beta} \neq \beta$.  The defender optimization is with respect to perceived values, not real values; the attacker perturbations mean that $p\left(D,x^t\right) = \hat{x}^t \neq x^t$ even though the attacker does not directly manipulate the state variables.  The dynamic attack problem is

\begin{gather}
\max_{\Delta \alpha^t} \sum \limits_{t=1}^{\tau} J(u^t,x^t,\theta) \\
\frac{1}{2} \sum_t \left\|\Delta \alpha^t\right\|_2^2 \leq \delta_{\alpha,max} \\
x^t = f(x^{t-1},u^t,\alpha^t,\beta) \\
\hat{x}^0 = x^0
\end{gather}

\noindent subject to

\begin{gather}
\min_{u^t} \sum \limits_{t=1}^{\tau} J(u^t,\hat{x}^t,\theta) \\
\hat{x}^t = f(\hat{x}^{t-1},u^t,\hat{\alpha}^t,\beta) \\
\hat{x}^{\tau} - \hat{x}^0 = 0 \\
g\left(\hat{x}^t,u^t,\hat{\alpha}^t,\beta\right)\leq 0 
\end{gather}

This, similarly, is a second level hypergame where $p\left(D,\alpha^t\right) = \hat{\alpha}^t \neq \alpha^t$.  As before, we could consider many variations on the dynamic and static attacks, but we will only look at these two scenarios here.

\subsection{Analytical Results}

The analytical results derived for the static optimization problem are applicable here as well.  If the dynamic optimization is convex, there are analogous results for perturbations to $\theta$, and it can similarly be shown that constraint perturbations (to $\beta$ and $\alpha^t$, in this case) cannot exhibit the same kind of local robustness as objective function perturbations.

\subsection{Test Problem}

Our MPC test problem is a single-zone HVAC system with a fan, heater, and chiller.  The objective is to minimize power consumption subject to physical constraints (e.g., the zonal temperature evolution) and operational constraints (e.g., remaining within comfort-defined temperature limits).  The baseline optimal control problem for the system is

\begin{gather}
\min \sum \limits_{t=1}^{\tau} \left[\theta_1 m^t + \theta_2 \left(m^t\right)^2 + \nu_h c_p m^t \left(T^t_i - d^t T^t_0 - \left(1-d^t\right) T^t_n\right) \right. \nonumber \\
\left. + \nu_n c_p m^t \left(T^t_{s,n} - T^t_s\right) + \nu_c c_p m^t \left(T^t_i - T^t_s\right)\right] \\
T^t_n = \left(1-\gamma\right) T^{t-1}_n + \beta m^t \left(T^t_{s,n} - T^t_n \right) + \gamma T^t_0 + Q^t_n \\
T^{\tau}_n - T^0_n = 0 \\
m_l \leq m^t \leq m_u \\
T^t_{s,n} - T^t_s \geq 0 \\
T^l_n \leq T^t_n \leq T^u_n \\
d_l \leq d^t \leq d_u \\
T^l_{s,n} \leq T^t_{s,n} \leq T^u_{s,n} \\
T^t_i - d^t T^t_0 - \left(1-d^t\right) T^t_n \geq 0 \\
T^t_i - T^t_s \geq 0 
\end{gather}

\noindent where $m^t$ is the mass flow rate, $T^t_i$ is the internal duct temperature, $T^t_s$ is the temperature of the air put out by the chiller, $T^t_{s,n}$ is the temperature of the air supplied to the zone, $T^t_n$ is the temperature of the zone, and $d^t$ is the damper position.  All of these are control variables.  $T^t_0$ is the external temperature (set to $25^{\circ}$C in this instantiation of the model); $\beta$ and $\gamma$ are scalar parameters that capture the room thermal properties.  Other quantities listed in the problem description are parameters that are not affected by any adversarial perturbations.  See Appendix \ref{HVAC Description} for more details.  The fan, heater, and chiller power consumption levels at each time step are

\begin{gather}
\theta_1 m^t + \theta_2 \left(m^t\right)^2 \\
\nu_h c_p m^t \left(T^t_i - d^t T^t_0 - \left(1-d^t\right) T^t_n\right) \nu_n c_p m^t \left(T^t_{s,n} - T^t_s\right) \\
\nu_c c_p m^t \left(T^t_i - T^t_s\right)
\end{gather}

\noindent respectively.  In this model, the static pressure is almost constant, and thus we omit it from the fan component of the model.  The static attack manipulates the defender perception of $\beta$ and $\gamma$.  The attacker goal is to maximize power consumption over the entire time horizon given that the defender observes $\hat{\beta} = \beta + \Delta \beta$ and $\hat{\gamma} = \gamma + \Delta \gamma$ and the attacker is constrained by

\begin{gather}
\frac{1}{2} \left[\left(\frac{\Delta \beta}{\beta} \right)^2 + \left(\frac{\Delta \gamma}{\gamma}\right)^2 \right] \leq \delta_{max}
\end{gather}

\noindent subject to the defender optimization of the original baseline problem.    Because $\beta$ and $\gamma$ are of different magnitudes, using relative perturbations, not absolute ones, avoids some potential problems.  We also highlight the previously mentioned differences between the perceived and actual state variables values.  For example, the true zone temperature, $T^t_n$, and the defender perception of the zone temperature, $\hat{T}^t_n$, will evolve according to the equations, respectively,

\begin{gather}
T^t_n = \left(1-\gamma\right) T^{t-1}_n + \beta m^t \left(T^t_{s,n} - T^t_n \right) + \gamma T^t_0 + Q^t_n \\
\hat{T}^t_n = \left(1-\hat{\gamma}\right) \hat{T}^{t-1}_n + \hat{\beta} m^t \left(T^t_{s,n} - \hat{T}^t_n \right) + \hat{\gamma} T^t_0 + Q^t_n 
\label{perceived thermal eqn}
\end{gather}

There will be a similar discrepancy between $T^t_i$ and $\hat{T}^t_i$.  The dynamic attack manipulates the defender's perception of $T^t_0$ so that $\hat{T}^t_0 = T^t_0 + \Delta T^t_0$ and $\frac{1}{2} \sum \limits_t \left(\Delta T^t_0\right)^2 \leq \Delta T_{max}$.  As in the static parameter manipulation case, the defender will misperceive both $T^t_n$ and $T^t_i$.  The full formulations for the static and dynamic manipulation problems are provided in Appendix \ref{HVAC Description}.

\section{Computational Implementation}

The specific calculations to turn each hypergame problem into a tractable nonlinear program (NLP) are provided in Appendices \ref{Fan Description} and \ref{HVAC Description}.  We summarize our general approach here.  Each hypergame produces a multi-level optimization problem.  To solve this, we write the optimality conditions of the lower level problems as complementarity conditions.  In the case of the fan optimization, we can transform these complementarity conditions into equality constraints and then solve the resulting problem as an NLP.  For the HVAC problem, we cannot do this, and this leaves us with a Mathematical Program with Equilibrium Constraints (MPEC) \cite{ruiz14jsr}.  We can solve the MPEC as a series of NLPs by relaxing the complementarity constraints and penalizing the relaxation with a progressively increasing weight.  For the work described in this paper, this was both reliable and efficient.  To implement our approach, we derived the necessary optimality conditions by hand, coded up the NLPs in MATLAB \cite{MATLAB}, and solved the NLPs using \textit{fmincon}.

\section{Results}
\subsection{Fan Optimization} 

\begin{table}[htp]
\caption{Objective Function Manipulation Results ($\delta_{\theta,max} = 0.1$)}
\label{tab:objfun results}
\centering
\setlength{\tabcolsep}{0.5em}
\begin{tabular}{c|cccccc}
Case												&$m$		&$p$		&$\Delta \theta_1$		&$\Delta \theta_2$		&$\Delta \theta_3$		&Power \\
\hline
Baseline										&2.06		&3.85		&-										&-										&-										&13.97 \\
True Manipulation						&2.02		&3.94		&0.150								&0.303								&0.292								&16.68 \\
Perception Manipulation			&2.29		&3.38		&-0.090								&-0.411								&0.151								&14.26 (12.42) \\
Faulty Defender Anticipation&1.95		&4.16		&-										&-										&-										&14.08 (14.71) \\
Double-Bluff Manipulation		&1.89		&4.42		&0.00684							&0.259								&-0.358								&14.30 (13.76)
\end{tabular}
\end{table}

Table \ref{tab:objfun results} shows the results for the attacker manipulation of the objective function parameters; power consumption values in parentheses indicate the power usage perceived by the defender where it differs from the actual usage.  Manipulating the true $\theta_i$ values produced a notable increase in power consumption compared with the baseline.  Manipulating defender perceptions, though, proved less effective.  For example, when the attacker manipulated the perceptions of an unsuspecting defender (Perception Manipulation), the gap between the perceived and actual power usage was noticeable, but the actual increase in power relative to the baseline case was small.  Similarly, if the defender erroneously thought that the attacker was manipulating the perceived values of $\theta_i$ (Faulty Defender Anticipation), the true power usage was almost identical to the baseline case, though the perceived power consumption was somewhat higher. 

When manipulating the defender's perceptions, the attacker got the defender to increase $m$ and decrease $p$ (relative to the baseline case) by decreasing the perceived value of $\theta_1$ and $\theta_2$ ($\Delta \theta_1 < 0$, $\Delta \theta_2 < 0$) and increasing the perceived value of $\theta_3$ ($\Delta \theta_3 > 0$).  This approach is more beneficial for the attacker than decreasing $m$ and increasing $p$ because the objective is quadratic in $m$ but only linear in $p$.  In the double-bluff situation, however, the defender expects the attacker to employ this optimal strategy, and so the attacker does the exact opposite (i.e., encourages the defender to increase $p$ and decrease $m$), which provides a slight additional benefit over the simple manipulation case.

\begin{figure}[htp]
\centering
\includegraphics[width=0.7\textwidth]{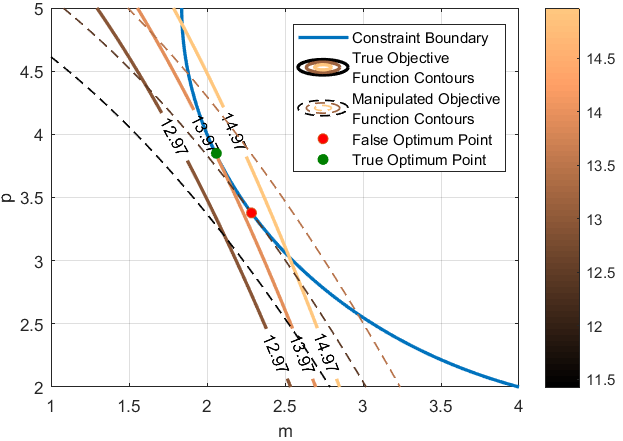}
\caption{Visualization of `Perception Manipulation' attack.}
\label{fig:manip plot}
\end{figure}

Fig. \ref{fig:manip plot} shows the `Perception Manipulation' case and why it produces so little payoff for the attacker.  There, we see how the perceived objective function contours are essentially a rotated version of the original objective function contours.  That rotation, produced by changes in the relative magnitudes of the $\theta_i$ parameters, produces a perceived (i.e., false) optimum point that is noticeably different from the true optimum point.  However, even a significant difference in the solution location does not necessarily translate to a large difference in the true objective function value because neither the constraint nor the objective function contours have large curvatures near the true optimum -- most of the translation between the two points is parallel to the contours of the true objective function.

Manipulating constraints gave the attacker more options than manipulating the objective function parameters.  As Table \ref{tab:cons results I} shows, constraint manipulation was also much more effective as an attacker strategy.  For example, when the attacker attempted to maximize power consumption against a defender who did not believe an attack was underway (Power Max, Normal), the attacker was able to increase power consumption by almost 30\% compared with the baseline.  Attempting to maximize the constraint violation (Break System, Normal) resulted in a significant level of violation, too.

\begin{table}[htp]
\caption{Constraint Manipulation Results ($\delta_{c,max} = 0.1$)}
\label{tab:cons results I}
\centering
\setlength{\tabcolsep}{0.5em}
\begin{tabular}{cc|cccc}
Attacker Action	&Defender Belief 		&$m$		&$p$		&Power 				&Violation\\
\hline
No Attack		&Normal									&2.06		&3.85		&13.97 				&-\\
Power Max		&Normal									&2.59		&4.22		&17.76 				&-\\
No Attack		&Power Max							&1.57		&3.37		&10.79 				&4.92\\
No Attack		&Break System						&2.59		&2.24		&17.76 				&-\\
Break System&Power Max							&1.17		&2.78		&8.11 				&4.85\\
Power Max		&Break System						&3.16		&4.53		&22.21 				&-\\
Break System&Normal									&1.58		&3.36		&10.79 				&2.20\\
Power Max (Double-Bluff)&Power Max 	&2.16		&3.94		&14.71 				&0.406\\
Break System (Double-Bluff)&Break System	&2.05		&3.87		&13.97				&0.003
\end{tabular}
\end{table}

\begin{table}[htp]
\caption{Constraint Manipulation Results ($\delta_{c,max} = 0.1$)}
\label{tab:cons results II}
\centering
\setlength{\tabcolsep}{0.5em}
\begin{tabular}{cc|ccc}
Attacker Action	&Defender Belief 		&$\Delta c_m$		&$\Delta c_p$		&$\Delta c_r$	\\
\hline
Power Max		&Normal									&0.301					&0.097					&0.316				\\
Break System&Power Max							&-0.285					&-0.137					&-0.316				\\
Power Max		&Break System						&0.301					&0.097					&0.316				\\
Break System&Normal									&-0.285					&-0.137					&-0.316				\\
Power Max (Double-Bluff)&Power Max 	&0.419					&0.157					&0.000				\\
Break System (Double-Bluff)&Break System	&-0.295		&-0.113					&-0.316				
\end{tabular}
\end{table}

In this case, there were also major consequences for wrongly anticipating an attack.  Anticipating a power maximization attack when there was no attack resulted in a worse constraint violation than when the attacker was deliberately trying to break the system.  Conversely, anticipating a `break system' attack when the actual attack was a `power max' attack led to an increase in power consumption of almost 60\% compared with the baseline.  Note that in these false anticipations, the attacker is assuming that the defender is just playing normally (i.e., the attacker is not taking advantage of the defender's mistake).  The double-bluff strategies did not provide much benefit to the attacker, though.

Table \ref{tab:cons results II} also shows the perturbations used by the attacker.  We can see that the attacker strategies for maximizing power consumption and breaking the system are almost exactly mirror opposites, which makes sense.  The double-bluff strategies are not that much different than the regular strategies that they correspond to, though, so it is not surprising that the double-bluff approach is not very effective.  Switching attack modes would be a better option if the defender is anticipating an attack, and though we did not calculate this here, it would be possible to calculate an optimal attack for one mode given that the defender is expecting the other mode.  Given how the two modes produce almost exactly opposite attacker strategies, the attacker strategy would likely be quite similar to the same attack mode employed against an unsuspecting defender. 

In general, changes in constraint parameters may result in larger objective function changes than changes in objective function parameters for two reasons.  Firstly, the changes in constraints will be multiplied by the dual variables (Lagrange or Kuhn-Tucker) associated with those constraints to produce a final change in the objective function.  Secondly, changing constraint values may result in the active set at the optimum also changing, and that could produce large, nonlinear changes in the objective function.  All in all, this likely makes constraint manipulation a much more attractive target for a would-be attacker than objective function manipulation.

\subsection{Single-Zone HVAC Control}

\begin{figure}[ht]
\centering
\includegraphics[width=0.6\textwidth]{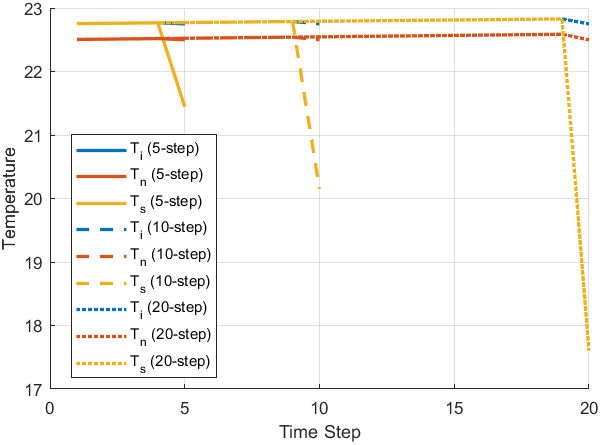}
\caption{Baseline temperature results.}
\label{fig:baseline}
\end{figure}

In the baseline case, and for all of the adversarial perturbations, $m^t$ and $d^t$ were both at their lower bounds for the entire optimization.  Fig. \ref{fig:baseline} shows the defender strategy in more detail for different optimization horizon lengths.  There, we see that the defender essentially allows the zone to evolve without manipulation until the last time step.  Because $T^t_0 > T^t_n$, this means that the zone warms over time, but because $\gamma$ is very small, this happens slowly.  At the last time step, the defender then chills the zone back to the initial temperature.  We can see this in the sudden drop in $T^t_s$ at the end of each time horizon; note that the optimization produces $T^t_s = T^t_{s,n}$ for each optimization.  This general behaviour is seen when the attacker manipulates defender perceptions, too.  The longer the optimization time horizon, the larger the drop in $T^t_s$ at the last time step.  If the length of the time horizon were increased sufficiently, eventually the system would require multiple steps of cooling, because $T^t_s$ would hit its lower bound.  $T^t_n$ never hit its upper bound, but if it did, this would also require additional cooling prior to the end of the optimization horizon.   

\begin{table}[ht]
\caption{Static Parameter Manipulation Results ($\delta_{max} = 0.1$)}
\label{tab:static parameter}
\centering
\setlength{\tabcolsep}{0.5em}
\begin{tabular}{c|ccc}
																&5-step								&10-step							&20-step \\
\hline
Baseline Power									&14.76								&29.48								&58.77 \\
Actual Power										&15.08								&30.27								&60.95 \\
Defender Perceived Power				&15.00								&29.97								&59.80	\\
$\Delta \beta$									&-1.81e-3							&-1.84e-3							&-1.94e-3 \\
$\Delta \gamma$									&1.64e-5							&1.52e-5							&9.74e-6 \\
$\lambda_{mean}$								&367									&370									&383
\end{tabular}
\end{table}

Table \ref{tab:static parameter} shows that manipulating the defender's perception of $\beta$ and $\gamma$ resulted in small power increases, relative to the baseline, and small discrepancies between the actual and perceived power use.  The perturbations themselves also change slightly as the length of the time horizon changes; there is a greater emphasis on $\Delta \beta$ as the time horizon gets longer.  In this model, $\beta$ essentially measures how hard it is to change the zone temperature with the HVAC system.  Setting $\Delta \beta < 0$ makes the defender think that the zone is harder to adjust than it actually is.  The $\gamma$ parameter then captures the heat transfer between the zone and the outside environment.  Setting $\Delta \gamma > 0$ makes the defender think that there is more heat transfer than there actually is.  All of this combines to increase the amount of cooling that the defender thinks is necessary at the end.  The $\Delta T$ plots in Figs. \ref{fig:percep static} and \ref{fig:percep dynamic} show this kind of behaviour: the defender thinks that the temperatures are higher than they actually are and therefore overcompensates at the end.  This overcompensation leads to an increase in power use and a final $T^t_n$ value  that is actually slightly lower than it should be.

Next, we can look at the $\lambda_{mean}$ values given in Table \ref{tab:static parameter}.  $\lambda_{mean}$ is the average of the Lagrange multipliers associated with (\ref{perceived thermal eqn}) and therefore provides a measure of how the $\Delta \beta$ and $\Delta \gamma$ perturbations get multiplied.  This value increases as the time horizon lengthens, which makes sense: as the time horizon lengthens, the importance of the thermal evolution process increases.  An attacker perturbing $\beta$ and $\gamma$ would want this value to be as large (positive or negative) as possible.

\begin{figure}[htp]
\centering
\begin{subfigure}[t]{0.6\textwidth}
\centering
\includegraphics[width=\textwidth]{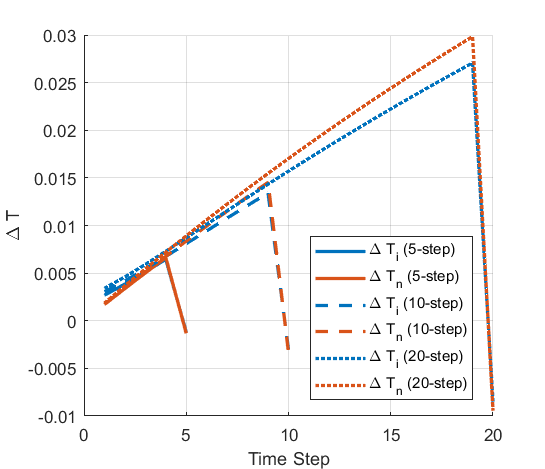}
\caption{Static attack.}
\label{fig:percep static}
\end{subfigure} \\
~\begin{subfigure}[t]{0.6\textwidth}
\centering
\includegraphics[width=\textwidth]{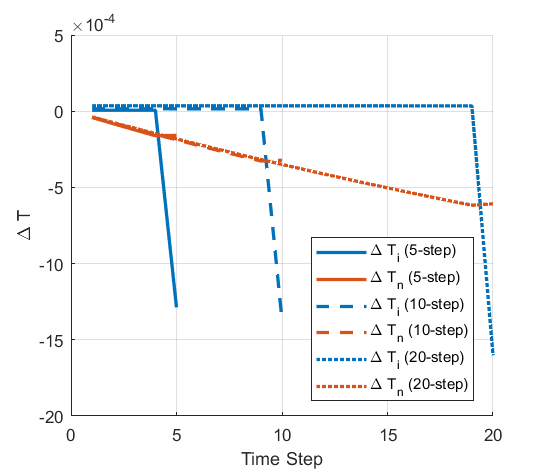}
\caption{Dynamic attack.}
\label{fig:percep dynamic}
\end{subfigure}
\caption{Temperature deviations, $\Delta T = \left(T_{true} - T_{perceived}\right)$.}
\end{figure}

\begin{table}[ht]
\caption{Dynamic Attack Results ($\Delta T_{max} = 0.1n$ for $n$-step problem)}
\label{tab:dynamic parameter}
\centering
\setlength{\tabcolsep}{0.5em}
\begin{tabular}{c|ccc}
																&5-step								&10-step							&20-step \\
\hline
Baseline Power									&14.76								&29.48								&58.77 \\
Actual Power										&16.35								&32.85								&65.68 \\
Defender Perceived Power				&15.58								&31.20								&62.27	\\
$\lambda_{mean}$								&219									&218									&216
\end{tabular}
\end{table}

\begin{figure}[htp]
\centering
\includegraphics[width=0.6\textwidth]{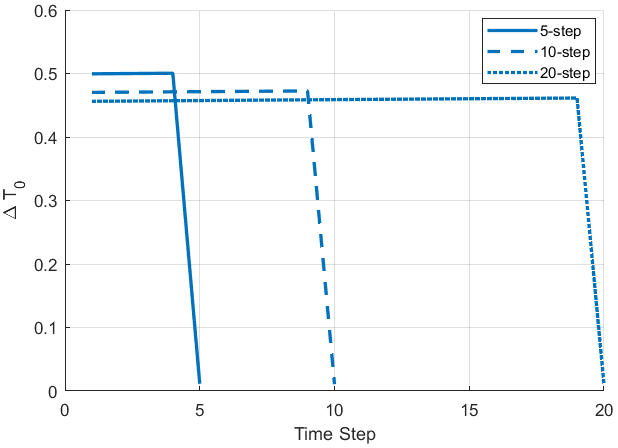}
\caption{Dynamic parameter manipulation temperature perturbations.}
\label{fig:percep2 perturb}
\end{figure}

\begin{table}[ht]
\caption{Power Consumption Comparisons relative to Baseline (\%)}
\label{tab:comparisons}
\centering
\setlength{\tabcolsep}{0.5em}
\begin{tabular}{c|ccc}
														&5-step							&10-step						&20-step \\
\hline
Static Attack (Perceived)		&1.6								&1.7								&1.8 \\
Static Attack (Actual)			&2.2								&2.7								&3.7	\\
Dynamic Attack (Perceived)	&5.6								&5.8								&6.0	\\
Dynamic Attack (Actual)			&10.1								&11.4								&11.8
\end{tabular}
\end{table}

Table \ref{tab:dynamic parameter} shows that manipulating $T^t_0$ provided a much larger increase in power consumption as well as a larger difference between the perceived and actual power consumption.  $\lambda_{mean}$ is also much smaller, and these phenomena are related.  The static parameters could only affect the power consumption indirectly through the temperature evolution equation.  $T^t_0$, however, shows up in the objective function and another constraint in addition to the temperature evolution equation, so increasing $\lambda_{mean}$ becomes less important.  In this case, misperceptions of $T^t_i$ and $T^t_n$ become smaller (see Figs. \ref{fig:percep static} and \ref{fig:percep dynamic}) and less important to the attacker.  Instead, the attacker uses $\Delta T^t_0 > 0$ to get the defender to increase $T^t_i$, and thus the defender ends up engaging the heater (because $T^t_i - d^t T^t_0 - \left(1-d^t\right) T^t_n > 0$ even though $\hat{T}^t_i - d^t \hat{T}^t_0 - \left(1-d^t\right) \hat{T}^t_n = 0$) as well as the chiller.  The perturbations themselves follow a clear pattern, as shown in Fig. \ref{fig:percep2 perturb}.  They increase very slightly over time until the last time step, at which point they drop to nearly zero.  The last step is less valuable to the attacker because there are no more thermal evolution steps left in the optimization at that point.  Table \ref{tab:comparisons} provides an overall summary of the power consumption results.  Generally speaking, the relative payoff for the attacker increases with the length of the time horizon.  The actual power consumed in the static attack scenario, relative to the baseline, is roughly proportional to the length of the time horizon, but the other three cases in Table \ref{tab:comparisons} all seem to plateau.

\section{Discussion}

\subsection{Stuxnet-like Attacks and Hypergames}

In this paper, we showed examples of how an attacker with knowledge of the system in question could manipulate the optimization processes of that system.  These problems were relatively small, but they were sufficient to show how the analysis works.  Hypergames are about strategic interactions when there are misperceptions and/or information asymmetries.  In this case, we were able to show how those asymmetries or misperceptions could affect system performance.  For example, getting the defender to respond to a non-existent threat could actually prove to be a very effective attacker strategy.  Conversely, it is possible for the defender system to have a natural robustness to perturbations (though that was not the case in these test problems).  We could consider more complex interactions, and we intend to do so in future work, but that future work will need to build upon the basics outlined here.

When we look at Stuxnet as a motivating example for this work, we can see that there are many similarities as well as some key differences between Stuxnet and the cases considered here.  In both Stuxnet and our case studies, the attacker employed limited deviations to avoid detection; we modelled this using the concept of an attacker budget.  Both also involved fake sensor signals ($\Delta T^t_0$) and manipulated calibration values ($\Delta \theta$, $\Delta c$, $\Delta \beta$, $\Delta \gamma$).  Our examples each had two different kinds of attack modes, and for the fan optimization, there were two different attack objectives for one of the modes, but these all involved negatively impacting the defender's control system in some way.  Finally, Stuxnet and the attacks considered in this paper all utilized deep knowledge of an automated decision-making system to determine how to perform the attack.

There are two primary sets of differences between this paper's case studies and Stuxnet.  Firstly, to the best of our knowledge, Stuxnet was not optimization-based, and the centrifuge control systems did not employ optimal control.  As such, the decision-making processes for both the attacker and the defender were different than in our paper.  Secondly, Stuxnet actually overrode the control signals and software to manipulate the centrifuges \cite{symantec11tr}, whereas our attacks only altered sensor and calibration data.  If we were trying to model the Stuxnet attack itself, these discrepancies would be problematic.  Given the more general nature of our investigation here, though, this is less of an issue.  Moreover, the key similarities identified above are ones we believe to be relevant to a wide range of control systems that might be threatened by cyber attacks in general and APTs in particular.

\subsection{Scalability Considerations}

A big question in applying these techniques to real-world problems is scalability.  These problems were relatively small; even the 20-step HVAC problem had only 120 variables (six per time step) in the baseline problem.  How easy would it be to propagate the optimality conditions and solve the resulting MPECs for larger systems?  The answer has two parts.  Firstly, there is the question of the optimality conditions.  If those optimality conditions are necessary but not sufficient, as in general continuous NLP problems, propagating the optimality conditions to turn the multi-level optimization into an MPEC may run into difficulties; multiple optima would be one example of this.  That being said, the single-zone HVAC system presented here was a nonconvex problem, and it had no such problems.  If there are more than two levels to the optimization, that can also cause difficulties, as the optimization conditions from lower levels compound.  This then leads into the question of tractability.  Adding the dual variables of lower level optimizations to the problem description in order to solve the system as an MPEC can greatly increase the number of variables involved; having multiple levels may exacerbate the issue.  However, it is sometimes possible to simplify the optimality conditions and thereby remove some of the dual variables (as was done for the fan optimization problem).  The NLP sequential relaxation of the MPEC also scales well and handles the complementarity constraints efficiently.  On the whole, the scalability of this approach will depend on the problem in question and how many levels of (mis)perception are of interest.  Hypergames where the individual players' games are differentiable, convex optimization problems are likely to have the greatest amount of success with this approach.  Problems with known or constant active constraint sets will also generally be more amenable to the multi-level optimizations than problems with active sets that change.

\subsection{Future Work}

Some authors writing on Stuxnet suggest the use of heuristics to identify attacks \cite{karnouskos11cp,bencsath12jsr}.  One area of future work would be to take existing research on learning in repeated hypergames \cite{takahashi99cp,gharesifard10cp} and apply it to this context.  For this, we would consider the defender's ability to detect attacks as well as the attacker's behaviour when the non-detection constraint is endogenous rather than exogenous; the attacker budget imposed here would be an example of an exogenous detection constraint.  Another area of interest would be the defender's decision-making more generally.  Given the possibility of attack and the potential consequences (as calculated in this paper), how should a defender respond if an attack is undetectable beforehand?  Hypergame results here should enable us to to evaluate and prescribe control policies more broadly.  Finally, we intend to extend this work to larger, real-world systems.  Working on such systems may then also involve more complicated attacker manipulations, but we anticipate being able to use the same techniques demonstrated here.

\section{Conclusions}
In this paper, we showed how hypergames can be extended to situations with continuous and time-varying variables.  That extension allowed us to consider the effects of adversarial perturbations in an optimal control context, which can give us insights into the control aspects of a Stuxnet-like attack.  Manipulating constraints can be a more effective attacker strategy than directly manipulating objective function parameters; our analytical results showed why we would expect this to be true more generally.  Moreover, the attacker need not change the underlying system in any way to attack successfully -- it may be sufficient to deceive the defender controlling the system.  It is possible to scale our approach up to larger systems, but the ability to do so will depend on the characteristics of the system in question, and we identified several characteristics that will make larger systems amenable to hypergame analysis.

%
% The next two lines define the bibliography style to be used, and the bibliography file.
\bibliographystyle{ieeetr}
\bibliography{bib}

% 
% If your work has an appendix, this is the place to put it.
\appendices

\section{Static Fan Optimization Calculations}
\label{Fan Description}

\subsection{Baseline Problem}

The baseline defender optimization is 

\begin{gather}
\min \limits_{m,p} \theta_1 m + \theta_2 m^2 + \theta_3 p \\
\frac{1}{2} \left[ \left(m-c_m\right)^2 + \left(p-c_p\right)^2 - c_r^2\right] \leq 0
\end{gather}

Note that we include the 1/2 factor in the constraint to cancel out factors of 2 that appear when taking the derivative of that constraint.  The objective function and inequality constraint are both convex functions, so the optimization is a convex problem and the KKT conditions are necessary and sufficient to define problem optima.  If we define the Lagrangian as $L$ and use $\lambda$ as the dual variable associated with the inequality constraint, we get the following optimality conditions:

\begin{gather}
\frac{\partial L}{\partial m} = \theta_1 + 2 \theta_2 m + \left(m-c_m\right) \lambda = 0 \\
\frac{\partial L}{\partial p} = \theta_3 + \left(p-c_p\right)\lambda = 0 \\
\frac{1}{2} \left[ \left(m-c_m\right)^2 + \left(p-c_p\right)^2 - c_r^2\right] \lambda = 0
\end{gather}

For these equations to be satisfied, $\lambda \neq 0$.  Since $\lambda \geq 0$, this ensures that $p < c_p$.  Moreover, if $c_r$ is sufficiently small, $m > 0$, and thus $m < c_m$.  We can then get rid of $\lambda$ by substitution, and we are left with

\begin{gather}
\left(p-c_p\right)\left( \theta_1 + 2 \theta_2 m\right) - \left(m-c_m\right) \theta_3 = 0 \\
\frac{1}{2} \left[ \left(m-c_m\right)^2 + \left(p-c_p\right)^2 - c_r^2\right] = 0
\end{gather}

\subsection{Objective Function Manipulation}
\label{Fan Objfun}

\subsubsection{Attacker Manipulates True/Physical Properties and Defender Knows}
\label{physical properties}

The min-max problem is

\begin{gather}
\min \limits_{m,p} \max \limits_{\Delta \theta_i} \left(\theta_1 + \Delta \theta_1\right) m + \left(\theta_2 + \Delta \theta_2\right) m^2 + \left(\theta_3 + \Delta \theta_3\right)p 
\end{gather}

\begin{gather}
\frac{1}{2} \left[ \left(m-c_m\right)^2 + \left(p-c_p\right)^2 - c_r^2\right] \leq 0 \\
\frac{1}{2} \sum_i \Delta \theta_i^2 \leq \delta_{\theta,max}
\end{gather}

We can use the attacker's KKT conditions to transform the min-max problem into a pure optimization problem.  Define $L$ as the Lagrangian and $\sigma$ as the dual variable associated with the attacker budget constraint.  Then

\begin{gather}
\frac{\partial L}{\partial \Delta \theta_1} = m - \sigma \Delta \theta_1 = 0  \Rightarrow \Delta \theta_1 = \frac{1}{\sigma} m \\
\frac{\partial L}{\partial \Delta \theta_2} = m^2 - \sigma \Delta \theta_2 = 0 \Rightarrow \Delta \theta_2 = \frac{1}{\sigma} m^2 \\
\frac{\partial L}{\partial \Delta \theta_3} = p - \sigma \Delta \theta_3 = 0 \Rightarrow \Delta \theta_1 = \frac{1}{\sigma} p
\end{gather}

For finite $\Delta \theta_i$, we require $\sigma \neq 0$.  Since we know, by definition, that $\sigma \geq 0$, then $\sigma > 0$.  We can therefore parameterize the attacker's decisions in terms of $\tau = 1/\sigma$:

\begin{gather}
\min \limits_{m,p} \max \limits_{\tau} \left(\theta_1 + m \tau \right) m + \left(\theta_2 + m^2 \tau \right) m^2 + \left(\theta_3 + p \tau \right)p \\
\frac{1}{2} \left[ \left(m-c_m\right)^2 + \left(p-c_p\right)^2 - c_r^2\right] \leq 0 \\
\frac{1}{2} \tau^2 \left(m^2 + m^4 + p^2\right) \leq \delta_{\theta,max}
\end{gather}

Given that the last constraint will always be active ($\sigma \neq 0$), we can solve for $\tau$:

\begin{gather}
\tau = \left[ \frac{2 \delta_{\theta,max}}{m^2 + m^4 + p^2} \right]^{\frac{1}{2}}
\end{gather}

We are then left with the following defender optimization:

\begin{gather}
\min \limits_{m,p} \theta_1 m + \theta_2 m^2 + \theta_3 p + \left[ 2 \delta_{\theta,max} \left(m^2 + m^4 + p^2\right) \right]^{\frac{1}{2}} \\
\frac{1}{2} \left[ \left(m-c_m\right)^2 + \left(p-c_p\right)^2 - c_r^2\right] \leq 0
\end{gather}

\subsubsection{Attacker Manipulates Defender Perceptions, Defender Unaware}
\label{perception, defender unaware}

The attacker is solving the problem

\begin{gather}
\max \limits_{\Delta \theta_i} \theta_1 m + \theta_2 m^2 + \theta_3 p \\
\frac{1}{2} \sum_i \Delta \theta_i^2 \leq \delta_{\theta,max}
\end{gather}

\noindent subject to the defender optimization

\begin{gather}
\min \limits_{m,p} \left(\theta_1 + \Delta \theta_1\right) m + \left(\theta_2 + \Delta \theta_2\right) m^2 + \left(\theta_3 + \Delta \theta_3\right)p \\
\frac{1}{2} \left[ \left(m-c_m\right)^2 + \left(p-c_p\right)^2 - c_r^2\right] \leq 0
\end{gather}

The optimality conditions of the defender problem are the same as in the baseline case except that we replace $\theta_i$ with $\hat{\theta}_i = \theta_i + \Delta \theta_i$:

\begin{gather}
\left(p-c_p\right)\left( \hat{\theta}_1 + 2 \hat{\theta}_2 m\right) - \left(m-c_m\right) \hat{\theta}_3 = 0 \\
\frac{1}{2} \left[ \left(m-c_m\right)^2 + \left(p-c_p\right)^2 - c_r^2\right] = 0
\end{gather}

This then results in the optimization problem for the attacker:

\begin{gather}
\max \limits_{\Delta \theta_i,m,p} \theta_1 m + \theta_2 m^2 + \theta_3 p \\
\frac{1}{2} \left[ \left(m-c_m\right)^2 + \left(p-c_p\right)^2 - c_r^2\right] = 0 \ \left(\rho\right) \\
\frac{1}{2} \sum_i \Delta \theta_i^2 \leq \delta_{\theta,max} \ \left(\mu\right) \\
\left(p-c_p\right)\left( \hat{\theta}_1 + 2 \hat{\theta}_2 m\right) - \left(m-c_m\right) \hat{\theta}_3 = 0 \ \left(\lambda\right)
\end{gather}

\noindent where the dual variable for each constraint is shown in brackets next to that constraint.  We can solve this directly as an optimization, but we can also use the optimality conditions to calculate $\Delta \theta_i$.  Define $L$ as the optimization's Lagrangian.  Then

\begin{gather}
\frac{\partial L}{\partial \Delta \theta_1} = -\mu \Delta \theta_1 + \left(p-c_p\right) \lambda = 0 \\
\frac{\partial L}{\partial \Delta \theta_2} = -\mu \Delta \theta_2 + 2 \left(p-c_p\right) m \lambda = 0 \\
\frac{\partial L}{\partial \Delta \theta_3} = -\mu \Delta \theta_3 - \left(m-c_m\right) \lambda = 0
\end{gather}

If we use $\tau = \lambda/\mu$, we get

\begin{gather}
\Delta \theta_1 = \tau \left(p-c_p\right)
\end{gather}

\begin{gather}
\Delta \theta_2 = 2 \tau \left(p-c_p\right) m \\
\Delta \theta_3 = - \tau \left(m-c_m\right) \\
\tau = \left[ \frac{2 \delta_{\theta,max}}{\left( p-c_p\right)^2 + \left( 2 \left(p-c_p\right) m\right)^2 + \left(m-c_m\right)^2} \right]^{\frac{1}{2}} = \left[ \frac{2 \delta_{\theta,max}}{4 \left(p-c_p\right)^2 m^2 + c_r^2} \right]^{\frac{1}{2}}
\end{gather}

We know that $\mu > 0$, but in principle $\lambda$ could be positive or negative.  When we solve the optimization directly (using the parameter values specified in the main body of the paper), we find that $\lambda > 0$.  Given that $p-c_p < 0$ and $m-c_m < 0$, this means that the attacker decreases the defender-perceived values of $\theta_1$ and $\theta_2$ while raising the defender-perceived value of $\theta_3$.  This in turn results in an increased value of $m$ and a decreased value of $p$ (relative to the unperturbed case).  The case where $\lambda < 0$ would correspond to the opposite behaviour.  

Both options produce local maxima, for the attacker, but in general, we would expect the $\lambda > 0$ option to produce a higher payoff: the objective is linear in $p$ but quadratic in $m$, so increasing $m$ would often provide a greater payoff than increasing $p$.  We do not have a proof delineating when this is the case, but we would expect this not to be the case only for small values of $\theta_1$ and $\theta_2$ (relative to $\theta_3$).  For the $c_m$, $c_p$, $c_r$, and $\delta_{\theta,max}$ values considered in this paper, we can empirically verify that for $\theta_1 \in [0.5,3.5]$, $\theta_2 \in [0.5,3.5]$, and $\theta_3 \in [0.5,3.5]$, the $\lambda >0$ option provides a larger attacker payoff.  This domain encompasses all of the true $\theta_i$ values that an attacker could manipulate to produce the $\hat{\theta}_i$ values observed by the defender.  Since the defender knows the attacker budget, if the defender believes that the attacker is attempting to perturb $\theta_i$, the defender can know that the attacker is employing the attack where $\tau > 0$.

\subsubsection{Attacker Manipulates Defender Perceptions, Defender is Aware}

Using the results from the previous section, the defender can reverse engineer the true $\theta_i$ values from the perceived values $\hat{\theta}_i$ if the defender is aware of an attack.  The defender believes that $\hat{\theta}_i$ has been calculated by an attacker solving the problem in Appendix \ref{perception, defender unaware}.  Therefore the defender's optimization is

\begin{gather}
\min \limits_{m,p}  \left(\hat{\theta}_1 - \Delta \theta_1\right) m + \left(\hat{\theta}_2 - \Delta \theta_2\right) m^2 + \left(\hat{\theta}_3 - \Delta \theta_3\right)p \\
\frac{1}{2} \left[ \left(m-c_m\right)^2 + \left(p-c_p\right)^2 - c_r^2\right] \leq 0 \\
\Delta \theta_1 = \tau \left(\hat{p} - c_p\right) \\
\Delta \theta_2 = 2 \tau \left(\hat{p} - c_p\right) \hat{m} \\
\Delta \theta_3 = - \tau \left(\hat{m} - c_m\right)  \\
\tau = \left[ \frac{2 \delta_{\theta,max}}{4 \left(\hat{p}-c_p\right)^2 \hat{m}^2 + c_r^2} \right]^{\frac{1}{2}} \\
\frac{1}{2} \left[ \left(\hat{m}-c_m\right)^2 + \left(\hat{p}-c_p\right)^2 - c_r^2\right] = 0 \\
\left(\hat{p} - c_p\right) \left(\hat{\theta}_1 + 2 \hat{\theta}_2 \hat{m}\right) -  \left(\hat{m} - c_m\right) \hat{\theta}_3 = 0
\end{gather}

\noindent where $\hat{m}$ and $\hat{p}$ are the decision variable values that the defender thinks that the attacker expects the defender to employ.  Note that it is possible to solve

\begin{gather}
\frac{1}{2} \left[ \left(\hat{m}-c_m\right)^2 + \left(\hat{p}-c_p\right)^2 - c_r^2\right] = 0 \\
\left(\hat{p} - c_p\right) \left(\hat{\theta}_1 + 2 \hat{\theta}_2 \hat{m}\right) -  \left(\hat{m} - c_m\right) \hat{\theta}_3 = 0
\end{gather}

\noindent once with the known $\hat{\theta}_i$ values and then use those to calculate $\Delta \theta_i$ -- these do not depend on $m$ or $p$.  Once this calculation has been performed, we are left with the original convex defender optimization problem.

\subsubsection{Attacker Manipulates Defender Perceptions, Defender is Aware, Attacker Knows that Defender is Aware}

This problem leads us to a multi-level optimization problem.  At level 1, we have the attacker optimization

\begin{gather}
\max \limits_{\Delta \theta_i} \theta_1 m + \theta_2 m^2 + \theta_3 p \\
\frac{1}{2} \sum_i \Delta \theta_i \leq \delta_{\theta,max} \\
\hat{\theta}_i = \theta_i + \Delta \theta_i
\end{gather}

At the next level (level 2), we have the defender optimization.  The defender performs his optimization based on the belief that the values he perceives, $\hat{\theta}_i$ has been perturbed by an attacker solving the problem in Appendix \ref{perception, defender unaware}.  Therefore the defender's optimization is

\begin{gather}
\min \limits_{m,p}  \left(\hat{\theta}_1 - \Delta \hat{\theta}_1\right) m + \left(\hat{\theta}_2 - \Delta \hat{\theta}_2\right) m^2 + \left(\hat{\theta}_3 - \Delta \hat{\theta}_3\right)p \\
\frac{1}{2} \left[ \left(m-c_m\right)^2 + \left(p-c_p\right)^2 - c_r^2\right] \leq 0 \\
\Delta \theta_1 = \tau \left(\hat{p} - c_p\right) \\
\Delta \theta_2 = 2 \tau \left(\hat{p} - c_p\right) \hat{m} \\
\Delta \theta_3 = - \tau \left(\hat{m} - c_m\right) \\
\tau = \left[ \frac{2 \delta_{\theta,max}}{4 \left(\hat{p}-c_p\right)^2 \hat{m}^2 + c_r^2} \right]^{\frac{1}{2}} \\
\frac{1}{2} \left[ \left(\hat{m}-c_m\right)^2 + \left(\hat{p}-c_p\right)^2 - c_r^2\right] = 0 \\
\left(\hat{p} - c_p\right) \left(\hat{\theta}_1 + 2 \hat{\theta}_2 \hat{m}\right) - \left(\hat{m} - c_m\right) \hat{\theta}_3 = 0
\end{gather}

The defender's optimality conditions (level 2) are then:

\begin{gather}
\frac{1}{2} \left[ \left(\hat{m}-c_m\right)^2 + \left(\hat{p}-c_p\right)^2 - c_r^2\right] = 0 \\
\left(\hat{p} - c_p\right) \left[\left(\theta_1 + \Delta \theta_1\right) + 2 \left(\theta_2 + \Delta \theta_2\right) \hat{m} \right] - \left(\hat{m} - c_m\right) \left(\theta_3 + \Delta \theta_3\right) = 0
\end{gather}

\begin{gather}
\frac{1}{2} \left[ \left(m-c_m\right)^2 + \left(p-c_p\right)^2 - c_r^2\right] = 0 \\
\left(p - c_p\right) \left[\left(\theta_1 + \Delta \theta_1 - \tau \left(\hat{p} - c_p\right) \right) + 2 \left(\theta_2 + \Delta \theta_2 - 2 \left(\hat{p} - c_p\right) \hat{m}\tau\right) m \right] \nonumber \\
- \left(m - c_m\right) \left(\theta_3 + \Delta \theta_3 + \tau \left(\hat{m} - c_m \right)\right) = 0 \\
\tau = \left[ \frac{2 \delta_{\theta,max}}{4 \left(\hat{p}-c_p\right)^2 \hat{m}^2 + c_r^2} \right]^{\frac{1}{2}}
\end{gather}

The attacker's optimization (level 1) is then

\begin{gather}
\max \limits_{\Delta \theta_i} \theta_1 m + \theta_2 m^2 + \theta_3 p \\
\frac{1}{2} \sum_i \Delta \theta_i \leq \delta_{\theta,max} \\
\frac{1}{2} \left[ \left(\hat{m}-c_m\right)^2 + \left(\hat{p}-c_p\right)^2 - c_r^2\right] = 0 \\
\left(\hat{p} - c_p\right) \left[\left(\theta_1 + \Delta \theta_1\right) + 2 \left(\theta_2 + \Delta \theta_2\right) \hat{m} \right] - \left(\hat{m} - c_m\right) \left(\theta_3 + \Delta \theta_3\right) = 0 \\
\frac{1}{2} \left[ \left(m-c_m\right)^2 + \left(p-c_p\right)^2 - c_r^2\right] = 0
\end{gather}

\begin{gather}
\left(p - c_p\right) \left[\left(\theta_1 + \Delta \theta_1 - \tau \left(\hat{p} - c_p\right) \right) + 2 \left(\theta_2 + \Delta \theta_2 - 2 \left(\hat{p} - c_p\right) \hat{m}\tau\right) m \right] \nonumber \\
- \left(m - c_m\right) \left(\theta_3 + \Delta \theta_3 + \tau \left(\hat{m} - c_m \right)\right) = 0 \\
\tau = \left[ \frac{2 \delta_{\theta,max}}{4 \left(\hat{p}-c_p\right)^2 \hat{m}^2 + c_r^2} \right]^{\frac{1}{2}}
\end{gather}

The attacker optimization may not be convex, but each $\Delta \theta_i$ value corresponds to a single $\left(\hat{m},\hat{p},m,p\right)$ tuple.  We can show by via a sequential analysis.  The equations

\begin{gather}
\frac{1}{2} \left[ \left(\hat{m}-c_m\right)^2 + \left(\hat{p}-c_p\right)^2 - c_r^2\right] = 0 \\
\left(\hat{p} - c_p\right) \left[\left(\theta_1 + \Delta \theta_1\right) + 2 \left(\theta_2 + \Delta \theta_2\right) \hat{m} \right] - \left(\hat{m} - c_m\right) \left(\theta_3 + \Delta \theta_3\right) = 0
\end{gather}

\noindent define a unique solution $\left(\hat{m},\hat{p}\right)$ to an instance of the unaware defender optimization.  By the logic employed in the previous section, we can calculate $\Delta \hat{\theta}_i$ values from that, which then in turn defines $m$ and $p$ as the unique solution to 

\begin{gather}
\frac{1}{2} \left[ \left(m-c_m\right)^2 + \left(p-c_p\right)^2 - c_r^2\right] = 0 \\
\left(p - c_p\right) \left[\left(\theta_1 + \Delta \theta_1 - \tau \left(\hat{p} - c_p\right) \right) + 2 \left(\theta_2 + \Delta \theta_2 - 2 \left(\hat{p} - c_p\right) \hat{m}\tau\right) m \right] \nonumber \\
- \left(m - c_m\right) \left(\theta_3 + \Delta \theta_3 + \tau \left(\hat{m} - c_m \right)\right) = 0 \\
\tau = \left[ \frac{2 \delta_{\theta,max}}{4 \left(\hat{p}-c_p\right)^2 \hat{m}^2 + c_r^2} \right]^{\frac{1}{2}}
\end{gather}

\subsection{Constraint Manipulation}
\label{Fan Constraint}

In this section, for the sake of simplicity, we assume that the attacker is only manipulating the constraint parameters (not the objective function parameters).  These constraint manipulations take the form of

\begin{gather}
\hat{c}_m = c_m + \Delta c_m \\
\hat{c}_p = c_p + \Delta c_p \\
\hat{c}_r = c_r - \Delta c_r
\end{gather}

The attacker is also subject to an attack budget of

\begin{equation}
\frac{1}{2} \left( \Delta c_m^2 + \Delta c_p^2 + \Delta c_r^2 \right) = \frac{1}{2} \sum_i \Delta c_i^2 \leq \delta_{c,max}
\end{equation}

\subsubsection{Attacker Manipulates Defender Perceptions, Defender Unaware}

 The attacker's optimization is

\begin{gather}
\max \limits_{\Delta c_i} \theta_1 m + \theta_2 m^2 + \theta_3 p \\
\frac{1}{2} \sum_i \Delta c_i^2 \leq \delta_{c,max}
\end{gather}

\noindent subject to the defender optimization

\begin{gather}
\min \limits_{m,p} \theta_1 m + \theta_2 m^2 + \theta_3 p \\
\frac{1}{2}\left[ \left(m - c_m - \Delta c_m \right)^2 + \left(p - c_p - \Delta c_p\right)^2 - \left(c_r - \Delta c_r\right)^2 \right] \leq 0
\end{gather}

The defender optimality conditions are

\begin{gather}
\left(p-c_p-\Delta c_p\right) \left(\theta_1 + 2 \theta_2 m\right) - \left(m-c_m-\Delta c_m\right)\theta_3 = 0 \\
\frac{1}{2}\left[ \left(m - c_m - \Delta c_m \right)^2 + \left(p - c_p - \Delta c_p\right)^2 - \left(c_r - \Delta c_r\right)^2 \right] = 0
\end{gather}

\noindent and we are left with the attacker optimization

\begin{gather}
\max \limits_{\delta_i} \theta_1 m + \theta_2 m^2 + \theta_3 p \\
\frac{1}{2} \sum_i \Delta c_i^2 \leq \delta_{c,max} \\
\left(p-c_p-\Delta c_p\right) \left(\theta_1 + 2 \theta_2 m\right) - \left(m-c_m-\Delta c_m\right)\theta_3 = 0 \\
\frac{1}{2}\left[ \left(m - c_m - \Delta c_m \right)^2 + \left(p - c_p - \Delta c_p\right)^2 - \left(c_r - \Delta c_r\right)^2 \right] = 0
\end{gather}

\subsubsection{Attacker Manipulates Defender Perceptions, Defender is Aware}

The defender's optimization is

\begin{gather}
\min \limits_{m,p} \theta_1 m + \theta_2 m^2 + \theta_3 p \\
\frac{1}{2}\left[ \left(m - \hat{c}_m + \Delta c_m \right)^2 + \left(p - \hat{c}_p + \Delta c_p\right)^2 - \left(\hat{c}_r + \Delta c_r\right)^2 \right] \leq 0
\end{gather}

\noindent where $\hat{c}_m$, $\hat{c}_p$, and $\hat{c}_r$ are the quantities that the defender perceives (which the defender believes to have been manipulated by the attacker).  The true parameter values are unknown, but the $\Delta c_i$ values are calculated by solving the attacker problem from the previous section:

\begin{gather}
\max \limits_{\hat{m},\hat{p},\Delta c_i} \theta_1 \hat{m} + \theta_2 \hat{m}^2 + \theta_3 \hat{p} \\
\frac{1}{2} \sum_i \Delta c_i^2 \leq \delta_{c,max} \ \left(\mu\right) \\
\left(\hat{p}-c_p-\Delta c_p\right) \left(\theta_1 + 2 \theta_2 \hat{m}\right) - \left(\hat{m}-c_m-\Delta c_m\right)\theta_3 = 0  \ \left(\sigma\right) \\
\frac{1}{2}\left[ \left(\hat{m} - c_m - \Delta c_m \right)^2 + \left(\hat{p} - c_p - \Delta c_p\right)^2 - \left(c_r - \Delta c_r\right)^2 \right] = 0 \ \left(\rho\right)
\end{gather}

\noindent where the dual variables for each constraint are shown in brackets beside the equation  Define $L$ as the Lagrangian for this problem.  The optimality conditions are then

\begin{gather}
\frac{\partial L}{\partial \hat{m}} = \theta_1 + 2 \theta_2 \hat{m} - \sigma \left( 2 \left(\hat{p} - c_p - \delta_p\right) \theta_2 - \theta_3 \right) - \rho \left(\hat{m} - c_m - \delta_m\right) = 0 \\
\frac{\partial L}{\partial \hat{p}} = \theta_3 - \sigma \left(\theta_1 + 2 \theta_2 \hat{m}\right) - \rho \left(\hat{p} - c_p - \delta_p\right) = 0 \\
\frac{\partial L}{\partial \delta_m} = -\mu \delta_m - \sigma \theta_3 + \rho \left(\hat{m} - c_m - \delta_m \right) = 0 \\
\frac{\partial L}{\partial \delta_p} = - \mu \delta_p + \sigma \left(\theta_1 + 2 \theta_2 \hat{m}\right) + \rho \left(\hat{p} - c_p - \delta_p\right) = 0 \\
\frac{\partial L}{\partial \delta_r} = -\mu \delta_r - \rho \left(c_r - \delta_r\right) = 0
\end{gather}

If we take the first two equations and simplify using $\hat{c}_i$, we get

\begin{gather}
\theta_1 + 2 \theta_2 \hat{m} - \sigma \left( 2 \left(\hat{p} - \hat{c}_p\right) \theta_2 - \theta_3 \right) - \rho \left(\hat{m} - \hat{c}_m\right) = 0 \\
\theta_3 - \sigma \left(\theta_1 + 2 \theta_2 \hat{m}\right) - \rho \left(\hat{p} - \hat{c}_p\right) = 0
\end{gather}

We can set this up to solve for $\sigma$ and $\rho$:

\begin{gather}
\left[ \begin{array}{cc}
2 \left(\hat{p} - \hat{c}_p\right) \theta_2 - \theta_3		&\hat{m} - \hat{c}_m \\
\theta_1 + 2 \theta_2 \hat{m}															&\hat{p} - \hat{c}_p \end{array} \right] 
\left\{ \begin{array}{c} \sigma \\ \rho \end{array} \right\} = \left\{ \begin{array}{c} \theta_1 + 2 \theta_2 \hat{m} \\ \theta_3 \end{array} \right\} %\\
%
%\left\{ \begin{array}{c} \sigma \\ \rho \end{array} \right\}  = \frac{1}{\left( 2 \left(\hat{p} - \hat{c}_p\right) \theta_2 - \theta_3\right) \left( \hat{p} - \hat{c}_p\right) - \left( \hat{m} - \hat{c}_m\right) \left( \theta_1 + 2 \theta_2 \hat{m}\right)} \nonumber \\
%\times \left[ \begin{array}{cc}
%\hat{p} - \hat{c}_p																				&-\left(\hat{m} - \hat{c}_m\right) \\
%-\left(\theta_1 + 2 \theta_2 \hat{m} \right)							&2 \left(\hat{p} - \hat{c}_p\right) \theta_2 - \theta_3 \end{array} \right] \left\{ \begin{array}{c} \theta_1 + 2 \theta_2 \hat{m} \\ \theta_3 \end{array} \right\}
\end{gather}

We can get closed-form expressions for $\sigma$ and $\rho$ by solving this 2x2 system analytically, and we can then use these expressions to calculate our $\Delta c_i$ values in terms of $\tau = 1/\mu$:

\begin{gather}
\Delta c_p = \tau \theta_3 \\
\Delta c_m = \tau \left[ \rho \left(\hat{m} - \hat{c}_m\right) - \sigma \theta_3\right] \\
\Delta c_r = -\tau \rho \hat{c}_r
\end{gather}

The constraint on the sum of squared $\Delta c_i$ values then lets us calculate a value for $\tau$:

\begin{gather}
\tau^2 \left[ \theta_3^2 + \left( \rho \left(\hat{m} - \hat{c}_m\right) - \sigma \theta_3\right)^2 + \rho^2 \hat{c}_r^2\right] = 2 \delta_{c,max} \\
\tau = \left[ \frac{2 \delta_{c,max}}{\theta_3^2 + \left[ \rho \left(\hat{m} - \hat{c}_m\right) - \sigma \theta_3\right]^2 + \rho^2 \hat{c}_r^2} \right]^{\frac{1}{2}}
\end{gather}

\noindent and thus we have closed-form expressions for the $\Delta c_i$ values that can then be plugged back into the original defender optimization without needing to know the true $c_i$ values.  Note that the defender can perform these calculations without knowing the true $c_i$ ahead of time -- it is sufficient to know $\hat{c}_i$.

\subsubsection{Attacker Manipulates Defender Perceptions, Defender is Aware, Attacker Knows that Defender is Aware}

The attacker's optimization is

\begin{gather}
\max \limits_{\Delta c_i} \theta_1 m + \theta_2 m^2 + \theta_3 p \\
\frac{1}{2} \sum_i \Delta c_i^2 \leq \delta_{c,max} \\
\hat{c}_m = c_m + \Delta c_m \\
\hat{c}_p = c_p + \Delta c_p \\
\hat{c}_r = c_r - \Delta c_r
\end{gather}

\noindent subject to the defender optimization from the previous section.  The optimality conditions for the defender's optimization are

\begin{gather}
\left(p-\hat{c}_p+\Delta \hat{c}_p\right)\left(\theta_1 + 2 \theta_2 m\right) - \left(m-\hat{c}_m+\Delta\hat{c}_m\right)\theta_3 = 0 \\
\frac{1}{2}\left[ \left(m - \hat{c}_m + \Delta \hat{c}_m \right)^2 + \left(p - \hat{c}_p + \Delta \hat{c}_p\right)^2 - \left(\hat{c}_r + \Delta \hat{c}_r\right)^2 \right] = 0
\end{gather}

\noindent where

\begin{gather}
\Delta \hat{c}_p = \tau \theta_3 \\
\Delta \hat{c}_m = \tau \left[ \rho \left(\hat{m} - \hat{c}_m\right) - \sigma \theta_3\right] \\
\Delta \hat{c}_r = -\tau \rho \hat{c}_r
\end{gather}

\begin{gather}
\tau = \left[ \frac{2 \delta_{c,max}}{\theta_3^2 + \left[ \rho \left(\hat{m} - \hat{c}_m\right) - \sigma \theta_3\right]^2 + \rho^2 \hat{c}_r^2} \right]^{\frac{1}{2}} \\
\left[ \begin{array}{cc}
2 \left(\hat{p} - \hat{c}_p\right) \theta_2 - \theta_3		&\hat{m} - \hat{c}_m \\
\theta_1 + 2 \theta_2 \hat{m}															&\hat{p} - \hat{c}_p \end{array} \right] 
\left\{ \begin{array}{c} \sigma \\ \rho \end{array} \right\} = \left\{ \begin{array}{c} \theta_1 + 2 \theta_2 \hat{m} \\ \theta_3 \end{array} \right\} \\
\left(\hat{p}-\hat{c}_p\right) \left(\theta_1 + 2 \theta_2 \hat{m}\right) - \left(\hat{m}-\hat{c}_m\right)\theta_3 = 0 \\
\frac{1}{2}\left[ \left(\hat{m} - \hat{c}_m \right)^2 + \left(\hat{p} - \hat{c}_p\right)^2 - \hat{c}_r^2 \right] = 0
\end{gather}

\subsubsection{Attacker Manipulates Defender to Break System, Defender is Unaware}

In this case, the attacker wants to cause the defender to deviate maximally from the constraint $\frac{1}{2} \left[\left(m-c_m\right)^2 + \left(p-c_p\right)^2 - c_r^2\right] \leq 0$ in the interest of causing a catastrophic failure.  The attacker's optimization is

\begin{gather}
\max \limits_{\Delta c_i} \frac{1}{2} \left[\left(m-c_m\right)^2 + \left(p-c_p\right)^2 - c_r^2\right] \\
\frac{1}{2} \sum_i \Delta c_i^2 \leq \delta_{c,max} \\
\left(p-c_p-\Delta c_p\right) \left(\theta_1 + 2 \theta_2 m\right) - \left(m-c_m-\Delta c_m\right)\theta_3 = 0 \\
\frac{1}{2}\left[ \left(m - c_m - \Delta c_m \right)^2 + \left(p - c_p - \Delta c_p\right)^2 - \left(c_r - \Delta c_r\right)^2 \right] = 0
\end{gather}

\subsubsection{Attacker Manipulates Defender to Break System, Defender Knows}

The defender's optimization is

\begin{gather}
\min \theta_1 m + \theta_2 m^2 + \theta_3 p \\
\frac{1}{2} \left[ \left(m-c_m\right)^2 + \left(p-c_p\right)^2 - c_r^2\right] \leq 0
\end{gather}

\noindent where the defender only observes $\hat{c}_i$ and needs to calculate $\Delta c_i$.  The defender knows that the attacker is solving the problem

\begin{gather}
\max \limits_{\Delta c_i} \frac{1}{2} \left[\left(\hat{m}-c_m\right)^2 + \left(\hat{p}-c_p\right)^2 - c_r^2\right] \\
\frac{1}{2} \sum_i \Delta c_i^2 \leq \delta_{c,max} \ \left(\mu\right) \\
\left(\hat{p}-c_p-\Delta c_p\right) \left(\theta_1 + 2 \theta_2 \hat{m}\right) - \left(\hat{m}-c_m-\Delta c_m\right)\theta_3 = 0 \ \left(\sigma\right) \\
\frac{1}{2}\left[ \left(\hat{m} - c_m - \Delta c_m \right)^2 + \left(\hat{p} - c_p - \Delta c_p\right)^2 - \left(c_r - \Delta c_r\right)^2 \right] = 0 \ \left(\rho\right)
\end{gather}

\noindent where the dual variables for each constraint are shown in brackets beside their respective equations.  If we define $L$ as the Lagrangian for that problem, the optimality conditions for this problem are

\begin{gather}
\frac{\partial L}{\partial \hat{m}} = \hat{m} - c_m + \sigma \left( 2 \theta_2 \left(\hat{p} - \hat{c}_p\right) + \theta_3\right) - \rho \left(\hat{m} - \hat{c}_m\right) = 0 \\
\frac{\partial L}{\partial \hat{p}} = \hat{p} - c_p - \sigma \left(\theta_1 + 2\theta_2 \hat{m}\right) - \rho \left(\hat{p} - \hat{c}_p\right) = 0 \\
\frac{\partial L}{\partial \Delta c_m} = -\mu \Delta c_m - \sigma \theta_3 + \rho \left(\hat{m} - \hat{c}_m\right) = 0 \\
\frac{\partial L}{\partial \Delta c_p} = - \mu \Delta c_p + \sigma \left(\theta_1 + 2 \theta_2 \hat{m}\right) + \rho \left(\hat{p} - \hat{c}_p\right) = 0 \\
\frac{\partial L}{\partial \Delta c_r} = - \mu \Delta c_r - \rho \hat{c}_r = 0
\end{gather}

We can solve for $\sigma$, $\rho$, and $\tau = 1/\mu$ to get expressions for $\Delta c_i$.  

\begin{gather}
\Delta c_m = \tau\left(\hat{m} - c_m + 2 \theta_2 \sigma \left(\hat{p} - \hat{c}_p \right) \right) \\
\Delta c_p = \tau \left(\hat{p} - c_p\right) \\
\Delta c_r = - \tau \rho \hat{c}_r \\
\left\{ \begin{array}{c} \sigma \\ \rho \end{array} \right\} = \frac{1}{- \theta_3 \left( \hat{p} - \hat{c}_p\right) - \left(\hat{m} - \hat{c}_m\right) \left(\theta_1 + 2 \theta_2 \hat{m} \right)}
\left[ \begin{array}{cc}
-\left(\hat{p} - \hat{c}_p\right)				&\hat{m} - \hat{c}_m \\
\theta_1 + 2 \theta_2 \hat{m}						&\theta_3 \end{array} \right] \left\{ \begin{array}{c} \hat{m} - c_m \\ \hat{p} - c_p \end{array} \right\} \\
\tau = \left[ \frac{2 \delta_{c,max}}{\left( \hat{m} - c_m + 2 \theta_2 \sigma \left(\hat{p} - \hat{c}_p \right) \right)^2 + \left(\hat{p} - c_p\right)^2 + \rho^2 \hat{c}_r^2 } \right]^{\frac{1}{2}}
\end{gather}

Unlike the result in the power maximization case, solving for $\Delta c_i$ requires knowing $c_i$, not just $\hat{c}_i$.  The defender then has to solve

\begin{gather}
\min \theta_1 m + \theta_2 m^2 + \theta_3 p \\
\frac{1}{2} \left[ \left(m-c_m\right)^2 + \left(p-c_p\right)^2 - c_r^2\right] \leq 0 \\
\hat{c}_m = c_m + \tau\left(\hat{m} - c_m + 2 \theta_2 \sigma \left(\hat{p} - \hat{c}_p \right) \right)  
\label{cm hat eqn} \\
\hat{c}_p = c_p + \tau \left(\hat{p} - c_p\right) 
\label{cp hat eqn} \\
\hat{c}_r = c_r + \tau \rho \hat{c}_r 
\label{cr hat eqn} \\
\left\{ \begin{array}{c} \sigma \\ \rho \end{array} \right\} = \frac{1}{- \theta_3 \left( \hat{p} - \hat{c}_p\right) - \left(\hat{m} - \hat{c}_m\right) \left(\theta_1 + 2 \theta_2 \hat{m} \right)}   \left[ \begin{array}{cc}
-\left(\hat{p} - \hat{c}_p\right)				&\hat{m} - \hat{c}_m \\
\theta_1 + 2 \theta_2 \hat{m}						&\theta_3 \end{array} \right] \left\{ \begin{array}{c} \hat{m} - c_m \\ \hat{p} - c_p \end{array} \right\} \\
\tau = \left[ \frac{2 \delta_{c,max}}{\left( \hat{m} - c_m + 2 \theta_2 \sigma \left(\hat{p} - \hat{c}_p \right) \right)^2 + \left(\hat{p} - c_p\right)^2 + \rho^2 \hat{c}_r^2 } \right]^{\frac{1}{2}} \\
\left(\theta_1 + 2 \theta_2 \hat{m}\right) \left(\hat{p}-\hat{c}_p\right) - \left(\hat{m}-\hat{c}_m\right)\theta_3 = 0 \\
\frac{1}{2}\left[ \left(\hat{m} - \hat{c}_m \right)^2 + \left(\hat{p} - \hat{c}_p\right)^2 - \hat{c}_r^2 \right] = 0
\end{gather}

\noindent where $\hat{c}_i$ is known.  This is actually less complicated than it appears, though.  We can calculate $\hat{m}$ and $\hat{p}$ only knowing $\theta_i$ and $\hat{c}_i$ (which are fixed) and using

\begin{gather}
\left(\theta_1 + 2 \theta_2 \hat{m}\right) \left(\hat{p}-\hat{c}_p\right) - \left(\hat{m}-\hat{c}_m\right)\theta_3 = 0 \\
\frac{1}{2}\left[ \left(\hat{m} - \hat{c}_m \right)^2 + \left(\hat{p} - \hat{c}_p\right)^2 - \hat{c}_r^2 \right] = 0
\end{gather}

With $\hat{m}$ and $\hat{p}$ known, $\sigma$ and $\rho$ are just linear functions of $c_i$, and we have another closed-form expression for $\tau$.  We are then left with three equations in three unknowns: solving (\ref{cm hat eqn})-(\ref{cr hat eqn}) for $c_i$.  These unknowns, moreover, do not depend on $m$ or $p$.

\subsection{Attacker Manipulates Defender to Break System, Defender Knows, Attacker Knows that Defender is Aware}

The attacker optimization is

\begin{gather}
\max \frac{1}{2} \left[ \left(m-c_m\right)^2 + \left(p-c_p\right)^2 - c_r^2\right] \\
\frac{1}{2} \sum_i \Delta c_i^2 \leq \delta_{c,max} \\
\hat{c}_m = c_m + \Delta c_m \\
\hat{c}_p = c_p + \Delta c_p \\
\hat{c}_r = c_r - \Delta c_r
\end{gather}

\noindent subject to the defender optimization

\begin{gather}
\min \theta_1 m + \theta_2 m^2 + \theta_3 p \\
\frac{1}{2} \left[ \left(m-\tilde{c}_m\right)^2 + \left(p-\tilde{c}_p\right)^2 - \tilde{c}_r^2\right] \leq 0
\end{gather}

\noindent where

\begin{gather}
\hat{c}_m = \tilde{c}_m + \tau\left(\hat{m} - \tilde{c}_m + 2 \theta_2 \sigma \left(\hat{p} - \hat{c}_p \right) \right) 
\label{tilde eqn start} \\
\hat{c}_p = \tilde{c}_p + \tau \left(\hat{p} - \tilde{c}_p\right) \\
\hat{c}_r = \tilde{c}_r + \tau \rho \hat{c}_r \\
\left\{ \begin{array}{c} \sigma \\ \rho \end{array} \right\} = \frac{1}{- \theta_3 \left( \hat{p} - \hat{c}_p\right) - \left(\hat{m} - \hat{c}_m\right) \left(\theta_1 + 2 \theta_2 \hat{m} \right)} \left[ \begin{array}{cc}
-\left(\hat{p} - \hat{c}_p\right)				&\hat{m} - \hat{c}_m \\
\theta_1 + 2 \theta_2 \hat{m}						&\theta_3 \end{array} \right] \left\{ \begin{array}{c} \hat{m} - c_m \\ \hat{p} - c_p \end{array} \right\} \\
\tau = \left[ \frac{2 \delta_{c,max}}{\left( \hat{m} - c_m + 2 \theta_2 \sigma \left(\hat{p} - \hat{c}_p \right) \right)^2 + \left(\hat{p} - c_p\right)^2 + \rho^2 \hat{c}_r^2 } \right]^{\frac{1}{2}} \\
\left(\theta_1 + 2 \theta_2 \hat{m}\right) \left(\hat{p}-\hat{c}_p\right) - \left(\hat{m}-\hat{c}_m\right)\theta_3 = 0 \\
\frac{1}{2}\left[ \left(\hat{m} - \hat{c}_m \right)^2 + \left(\hat{p} - \hat{c}_p\right)^2 - \hat{c}_r^2 \right] = 0
\label{tilde eqn end}
\end{gather}

The quantities with tildes on them indicate that these values are what the defender \emph{believes} to be the true values.  Given that (\ref{tilde eqn start})-(\ref{tilde eqn end}) not depend on $m$ or $p$, the defender optimality conditions are

\begin{gather}
\left(\theta_1 + 2 \theta_2 m\right) \left(p-\tilde{c}_p\right) - \left(m-\tilde{c}_m\right)\theta_3 = 0 \\
\frac{1}{2}\left[ \left(m - \tilde{c}_m \right)^2 + \left(p - \tilde{c}_p\right)^2 - \tilde{c}_r^2 \right] = 0
\end{gather}

The full attacker optimization is then

\begin{gather}
\max \frac{1}{2} \left[ \left(m-c_m\right)^2 + \left(p-c_p\right)^2 - c_r^2\right] \\
\frac{1}{2} \sum_i \Delta c_i^2 \leq \delta_{c,max} \\
\hat{c}_m = c_m + \Delta c_m \\
\hat{c}_p = c_p + \Delta c_p \\
\hat{c}_r = c_r - \Delta c_r \\
\hat{c}_m = \tilde{c}_m + \tau\left(\hat{m} - \tilde{c}_m + 2 \theta_2 \sigma \left(\hat{p} - \hat{c}_p \right) \right)  \\
\hat{c}_p = \tilde{c}_p + \tau \left(\hat{p} - \tilde{c}_p\right) \\
\hat{c}_r = \tilde{c}_r + \tau \rho \hat{c}_r \\
\left\{ \begin{array}{c} \sigma \\ \rho \end{array} \right\} = \frac{1}{- \theta_3 \left( \hat{p} - \hat{c}_p\right) - \left(\hat{m} - \hat{c}_m\right) \left(\theta_1 + 2 \theta_2 \hat{m} \right)} \left[ \begin{array}{cc}
-\left(\hat{p} - \hat{c}_p\right)				&\hat{m} - \hat{c}_m \\
\theta_1 + 2 \theta_2 \hat{m}						&\theta_3 \end{array} \right] \left\{ \begin{array}{c} \hat{m} - c_m \\ \hat{p} - c_p \end{array} \right\}
\end{gather}

\begin{gather}
\tau = \left[ \frac{2 \delta_{c,max}}{\left( \hat{m} - c_m + 2 \theta_2 \sigma \left(\hat{p} - \hat{c}_p \right) \right)^2 + \left(\hat{p} - c_p\right)^2 + \rho^2 \hat{c}_r^2 } \right]^{\frac{1}{2}} \\
\left(\theta_1 + 2 \theta_2 \hat{m}\right) \left(\hat{p}-\hat{c}_p\right) - \left(\hat{m}-\hat{c}_m\right)\theta_3 = 0 \\
\frac{1}{2}\left[ \left(\hat{m} - \hat{c}_m \right)^2 + \left(\hat{p} - \hat{c}_p\right)^2 - \hat{c}_r^2 \right] = 0
\end{gather}

\section{Single-Zone HVAC Control Calculations}
\label{HVAC Description}

\subsection{Baseline Problem}

The baseline problem is a power minimization problem for a heater, chiller, and fan together affecting a single zone of interest:

\begin{gather}
\min \sum \limits_{t=1}^{\tau} \left[\theta_1 m^t + \theta_2 \left(m^t\right)^2 + \nu_h c_p m^t \left(T^t_i - d^t T^t_0 - \left(1-d^t\right) T^t_n\right) \right. \nonumber \\
\left. + c_p \nu_n m^t \left(T^t_{s,n} - T^t_s\right) + \nu_c c_p m^t \left(T^t_i - T^t_s\right)\right] \\
- T^t_n + \left(1-\gamma\right) T^{t-1}_n + \beta m^t \left(T^t_{s,n} - T^t_n \right) + \gamma T^t_0 + Q^t_n = 0 \ \left(\lambda^t\right) \\
T^{\tau}_n - T^0_n = 0 \ \left(\mu_{\tau}\right) \\
m^t - m_l \geq 0 \ \left(\sigma^t_{m,l}\right) \\
m_u - m^t \geq 0 \ \left(\sigma^t_{m,u}\right) \\
T^t_{s,n} - T^t_s \geq 0 \ \left(\sigma^t_s\right) \\
T^t_n - T^l_n \geq 0 \ \left(\sigma^t_l\right) \\
T^u_n - T^t_n \geq \ \left(\sigma^t_u\right) \\
d^t - d_l \geq 0 \ \left(\sigma^t_{d,l}\right) \\
d_u - d^t \geq 0 \ \left(\sigma^t_{d,u}\right) \\
T^t_{s,n} - T^l_{s,n} \geq 0 \ \left(\sigma^t_{snl}\right) \\
T^u_{s,n} - T^t_{s,n} \geq 0 \ \left(\sigma^t_{snu}\right) \\
T^t_i - d^t T^t_0 - \left(1-d^t\right) T^t_n \geq 0 \ \left(\sigma^t_{in}\right) \\
T^t_i - T^t_s \geq 0 \ \left(\sigma^t_{is}\right)
\end{gather}

\noindent where the quantities in brackets after each equation are the dual variables corresponding to those equations.  Descriptions of the model variables and the model parameters are given in Tables \ref{HVAC Variables} and \ref{HVAC Parameters}, respectively.  This is a single-zone version of a multi-zone HVAC model.  The goal of the system is to manage the temperature in that single zone.  To do this, it takes in a mixture of air from the zone and from the environment, heats that air (if necessary) at a central heating unit, cools the air (if necessary) with a chiller, and uses a fan to send the air through HVAC ducting.  In a multi-zone model, there would be a local heater for each zone to provide any zone-specific heating; for our single-zone model, we retain the local heater in the interest of maintaining the same model structure.

\begin{table}[htp]
\centering
\caption{HVAC Control Variables}
\label{HVAC Variables}
\vspace{6pt}
\begin{tabular}{c|l}
Quantity										&Description \\
\hline
$m^t$												&Mass flow rate \\
$T^t_i$											&Temperature of air put out by central heating unit \\
$d^t$												&Fraction of HVAC input air coming from environment \\
$T^t_n$											&Zone temperature \\
$T^t_{s,n}$									&Temperature of air supplied to zone \\
$T^t_s$											&Output air temperature of chiller \\
\end{tabular}
\end{table}

\begin{table}[htp]
\centering
\caption{HVAC Model Parameters}
\label{HVAC Parameters}
\vspace{6pt}
\begin{tabular}{c|c|l}
Quantity										&Value 			&Description \\
\hline
$\theta_1$									&0.1				&Fan power consumption parameter \\
$\theta_2$									&0.1				&Fan power consumption parameter \\
$\nu_h$,$\nu_n$,$\nu_c$			&0.99				&Heater and chiller efficiencies \\
$c_p$												&1					&Specific heat of air \\
$T^t_0$											&25					&Environment air temperature at time $t$ \\
$\beta$											&0.0045			&Parameter describing temperature evolution \\
$\gamma$										&8.4e-6			&Parameter describing temperature evolution \\
$Q^t_n$											&0					&Thermal load at time $t$ \\
$\tau$											&varies			&Length of optimization horizon \\
$d_l$,$d_u$									&0.2, 0.5		&Lower and upper bounds on $d^t$ \\
$m_l$,$m_u$									&3.93, 13.1	&Lower and upper bounds on $m^t$ \\
$T^l_n$,$T^u_n$							&21.1, 23.9	&Lower and upper bounds on $T^t_n$ \\
$T^l_{s,n}$,$T^u_{s,n}$			&12.7, 35		&Lower and upper bounds on $T^t_{s,n}$
\end{tabular}
\end{table}

All of the other parameters with $l$ or $u$ in them correspond to lower or upper bounds on their respective variables.

At each time step $t$, the fan consumes power $\theta_1 m^t + \theta_2 \left(m^t\right)^2$ to move air through the system, the chiller consumes power $\nu_c c_p m^t \left(T^t_i - T^t_s\right)$, and the central heating unit consumes power $\nu_h c_p m^t \left(T^t_i - d^t T^t_0 - \left(1-d^t\right) T^t_n\right)$ and the zonal heater consumes power $c_p \nu_n m^t \left(T^t_{s,n} - T^t_s\right)$.  Most of the constraints are variable upper and lower bounds or physical constraints on the system (e.g., the temperature evolution of the room, the heater outputting air that is at least as warm as the air it takes in).  However, there is an endpoint constraint $T^{\tau}_n = T^0_n$ that is essentially a design constraint: at the end of the optimization horizon, the zone needs to be at the same temperature it was at the beginning of the horizon.  If we define the Lagrangian for this problem as $L$, the optimality conditions for this problem are

\begin{gather}
\frac{\partial L}{\partial m^t} = \theta_1 + 2 \theta_2 m^t + \nu_h c_p \left(T^t_i - d^t T^t_0 - \left(1-d^t\right) T^t_n\right) + c_p \nu_n \left(T^t_{s,n} - T^t_s\right) + \nu_c c_p \left(T^t_i - T^t_s\right) \nonumber \\
+ \lambda^t \beta \left(T^t_{s,n} - T^t_n \right) + \sigma^t_{m,u} - \sigma_{m,l}^t = 0 \\
\frac{\partial L}{\partial d^t} = \nu_h c_p m^t \left(T^t_n - T^t_0\right) + \sigma^t_{d,u} - \sigma^t_{d,l} - \sigma^t_{in} \left(T^t_n - T^t_0\right) = 0 \\
\frac{\partial L}{\partial T^t_n} = \nu_h c_p m^t \left(d^t - 1\right) + \lambda^t \left(-1 - \beta m^t\right) - \delta_{t\tau} \mu_{\tau} \nonumber \\
+ \left(1-\gamma\right) \lambda^{t+1} - \sigma^t_{in} \left(d^t-1\right) - \sigma^t_l + \sigma^t_u = 0 \\
\frac{\partial L}{\partial T^t_{s,n}} = c_p \nu_n m^t + \lambda^t \beta m^t - \sigma^t_s - \sigma^t_{snl} + \sigma^t_{snu} = 0 \\
\frac{\partial L}{\partial T^t_s} = -c_p \nu_n m^t - \nu_c c_p m^t + \sigma^t_s + \sigma^t_{is} = 0 \\
\frac{\partial L}{\partial T^t_i} = \nu_h c_p m^t + \nu_c c_p m^t - \sigma^t_{in} - \sigma^t_{is} = 0
\end{gather}

\noindent plus the optimization problem constraints listed above; note that $\delta_{t\tau}$, is a Kronecker delta, so it is 1 if $t=\tau$ and 0 otherwise.  These derivative conditions can simplify down to

\begin{gather}
0 \leq \theta_1 + 2 \theta_2 m^t +  \nu_h c_p \left(T^t_i - d^t T^t_0 - \left(1-d^t\right) T^t_n\right) + c_p \nu_n \left(T^t_{s,n} - T^t_s\right) + \nu_c c_p \left(T^t_i - T^t_s\right) \nonumber \\
+ \lambda^t \beta \left(T^t_{s,n} - T^t_n \right) + \sigma^t_{m,u} \perp  m^t - m_l \geq 0 \\
0 \leq \sigma^t_{m,u} \perp m_u - m^t \geq 0 \\
0 \leq d^t - d_l \perp \left(\sigma_{is} - \nu_c c_p m^t\right) \left(T^t_n - T^t_0\right) + \sigma^t_{d,u} \geq 0 \\
0 \leq \sigma^t_{d,u} \perp d_u - d^t \geq 0 \\
0 \leq \left(\sigma_{is} - \nu_c c_p m^t\right) \left(d^t - 1\right) - \lambda^t \left(1 + \beta m^t\right) -\delta_{t\tau} \mu_{\tau} + \left(1-\gamma\right) \lambda^{t+1} + \sigma^t_u \perp T^t_n - T^l_n \geq 0 \\
0 \leq \sigma^t_u \perp T^u_n - T^t_n \geq 0 \\
0 \leq\lambda^t \beta m^t + \sigma_{is} - \nu_c c_p m^t  + \sigma^t_{snu} \perp T^t_{s,n} - T^l_{s,n} \geq 0 \\
0 \leq \sigma^t_{snu} \perp T^u_{s,n} - T^t_{s,n} \geq 0 \\
0 \leq \nu_h c_p m^t - \left(\sigma_{is} - \nu_c c_p m^t\right) \perp T^t_i - d^t T^t_0 - \left(1-d^t\right) T^t_n \geq 0 \\
0 \leq \nu_n c_p m^t - \left(\sigma_{is} - \nu_c c_p m^t\right) \perp T^t_{s,n} - T^t_s \geq 0 \\
0 \leq \sigma^t_{is} \perp T^t_i - T^t_s \geq 0
\end{gather}

\noindent where $x \perp y$ indicate the complementarity constraint $xy=0$.  In general, this problem is nonconvex.  However, the parameter values specified above result in $m^t=m_l$ and $d^t=d_l$ for all $t$.  If we take these variables as constants, then the objective function and constraints are all linear in the model variables, so the optimization is a linear program, and the optimality conditions are then necessary and sufficient.  More generally, as long as the fan consumes most of the power (as it does in this case), it will be advantageous to keep $m^t$ as small as possible, and as long as the environment temperature differs from the zone temperature, the controller will always be incentivized to minimize the amount of outside air brought in (air that will have to be heated or cooled to reach the zone temperature).

\subsection{Attacker Manipulates Defender Perceptions of Static Parameters}

The attacker can manipulate the defender's perception of $\beta$ and $\gamma$ to maximize power consumption over the entire time horizon:

\begin{gather}
\max \sum \limits_{t=1}^{\tau} \left[\theta_1 m^t + 2 \theta_2 \left(m^t\right)^2 + \nu_h c_p m^t \left(T^t_i - d^t T^t_0 - \left(1-d^t\right) T^t_n\right) \right. \nonumber \\
\left. + c_p \nu_n m^t \left(T^t_{s,n} - T^t_s\right) + \nu_c c_p m^t \left(T^t_i - T^t_s\right)\right] \\
T^t_n = \left(1-\gamma\right) T^{t-1}_n + \beta m^t \left(T^t_{s,n} - T^t_n \right) + \gamma T^t_0 + Q^t_n \\
\hat{\beta} = \beta + \Delta \beta \\
\hat{\gamma} = \gamma + \Delta \gamma
\end{gather}

\begin{gather}
\frac{1}{2} \left[\left(\frac{\Delta \beta}{\beta} \right)^2 + \left(\frac{\Delta \gamma}{\gamma}\right)^2 \right] - \delta_{max} \leq 0 \\
0 \leq T^t_i - T^t_s \perp T^t_i - d^t T^t_0 - \left(1-d^t\right) T^t_n \geq 0
\end{gather}

\noindent subject to the defender optimality conditions

\begin{gather}
0 \leq \theta_1 + 2 \theta_2 m^t +  \nu_h c_p \left(\hat{T}^t_i - d^t T^t_0 - \left(1-d^t\right) \hat{T}^t_n\right) + c_p \nu_n \left(T^t_{s,n} - T^t_s\right) + \nu_c c_p \left(\hat{T}^t_i - T^t_s\right) \nonumber \\
+ \lambda^t \hat{\beta} \left(T^t_{s,n} - \hat{T}^t_n \right) + \sigma^t_{m,u} \perp  m^t - m_l \geq 0 \\
0 \leq \sigma^t_{m,u} \perp m_u - m^t \geq 0 \\
0 \leq d^t - d_l \perp \left(\sigma_{is} - \nu_c c_p m^t\right) \left(\hat{T}^t_n - T^t_0\right) + \sigma^t_{d,u} \geq 0 \\
0 \leq \sigma^t_{d,u} \perp d_u - d^t \geq 0 \\
0 \leq \left(\sigma_{is} - \nu_c c_p m^t\right) \left(d^t - 1\right) - \lambda^t \left(1 + \hat{\beta} m^t\right) -\delta_{t\tau} \mu_{\tau} + \left(1-\hat{\gamma}\right) \lambda^{t+1} + \sigma^t_u \perp \hat{T}^t_n - T^l_n \geq 0 \\
0 \leq \sigma^t_u \perp T^u_n - \hat{T}^t_n \geq 0 \\
0 \leq\lambda^t \hat{\beta} m^t + \sigma_{is} - \nu_c c_p m^t  + \sigma^t_{snu} \perp T^t_{s,n} - T^l_{s,n} \geq 0 \\
0 \leq \sigma^t_{snu} \perp T^u_{s,n} - T^t_{s,n} \geq 0 \\
0 \leq \nu_h c_p m^t - \left(\sigma_{is} - \nu_c c_p m^t\right) \perp \hat{T}^t_i - d^t T^t_0 - \left(1-d^t\right) \hat{T}^t_n \geq 0 \\
0 \leq \nu_n c_p m^t - \left(\sigma_{is} - \nu_c c_p m^t\right) \perp T^t_{s,n} - T^t_s \geq 0 \\
0 \leq \sigma^t_{is} \perp \hat{T}^t_i - T^t_s \geq 0 \\
- \hat{T}^t_n + \left(1-\hat{\gamma}\right) \hat{T}^{t-1}_n + \hat{\beta} m^t \left(T^t_{s,n} - \hat{T}^t_n \right) + \hat{\gamma} T^t_0 + Q^t_n = 0 \\
\hat{T}^T_n - T^0_n = 0 
\end{gather}

Note that the defender conditions are with respect to perceived/perturbed values, not real values (hence the $\hat{}$ on certain quantities).  The defender directly controls most of the variables (e.g., $m^t$, $T^t_s$) but does not directly control $T^t_i$ or $T^t_n$.  These variables are essentially functions of processes governed by other variables.  As such, $\hat{T}^t_i$ and $\hat{T}^t_n$ are the defender's perceived values for these variables.  The true equations governing the evolution of $T^t_n$ and $T^t_i$ are, respectively,

\begin{gather}
T^t_n = \left(1-\gamma\right) T^{t-1}_n + \beta m^t \left(T^t_{s,n} - T^t_n \right) + \gamma T^t_0 + Q^t_n \\
0 \leq T^t_i - T^t_s \perp T^t_i - d^t T^t_0 - \left(1-d^t\right) T^t_n
\end{gather}

The complementarity constraint ensures that $T^t_i$ is the minimum of $T^t_s$ and $d^t T^t_0 + \left(1-d^t\right) T^t_n$.  If $T^t_i > T^t_s$, the defender spends energy to cool the air and if $T^t_i > d^t T^t_0 + \left(1-d^t\right) T^t_n$, the defender spends energy to heat the air.

\subsection{Attacker Manipulates Defender Perceptions of Time-Varying Parameters}

The attacker can also manipulate the defender's perception of $T^t_0$ to maximize power consumption over the entire time horizon:

\begin{gather}
\max \limits_{\Delta T^t_0} \sum_t \left[ \theta_1 m^t + \theta_2 \left(m^t\right)^2 + \nu_h c_p m^t \left(T^t_i - d^t T^t_0 - \left(1-d^t\right) T^t_n\right) \right. \nonumber \\
\left. + c_p \nu_n m^t \left(T^t_{s,n} - T^t_s\right) + \nu_c c_p m^t \left(T^t_i - T^t_s\right)\right] \\
\frac{1}{2} \sum_t \left(\Delta T^t_0\right)^2 \leq \Delta T_{max} \\
\hat{T}^t_0 = T^t_0 + \Delta T^t_0 \\
- T^t_n + \left(1-\gamma\right) T^{t-1}_n + \beta m^t \left(T^t_{s,n} - T^t_n \right) + \gamma T^t_0 + Q^t_n = 0 \\
0 \leq T^t_i - T^t_s \perp T^t_i - d^t T^t_0 - \left(1-d^t\right) T^t_n \geq 0
\end{gather}

\noindent subject to the defender optimality conditions

\begin{gather}
\hat{T}^T_n - T^0_n = 0 \\
0 \leq \theta_1 + 2 \theta_2 m^t +  \nu_h c_p \left(\hat{T}^t_i - d^t \hat{T}^t_0 - \left(1-d^t\right) \hat{T}^t_n\right) + c_p \nu_n \left(T^t_{s,n} - T^t_s\right) + \nu_c c_p \left(\hat{T}^t_i - T^t_s\right) \nonumber \\
+ \lambda^t \beta \left(T^t_{s,n} - \hat{T}^t_n \right) + \sigma^t_{m,u} \perp  m^t - m_l \geq 0 \\
0 \leq \sigma^t_{m,u} \perp m_u - m^t \geq 0 \\
0 \leq d^t - d_l \perp \left(\sigma_{is} - \nu_c c_p m^t\right) \left(\hat{T}^t_n - \hat{T}^t_0\right) + \sigma^t_{d,u} \geq 0 \\
0 \leq \sigma^t_{d,u} \perp d_u - d^t \geq 0 \\
0 \leq \left(\sigma_{is} - \nu_c c_p m^t\right) \left(d^t - 1\right) - \lambda^t \left(1 + \beta m^t\right) -\delta_{t\tau} \mu_{\tau} + \left(1-\gamma\right) \lambda^{t+1} + \sigma^t_u \perp \hat{T}^t_n - T^l_n \geq 0 \\
0 \leq \sigma^t_u \perp T^u_n - \hat{T}^t_n \geq 0 \\
0 \leq\lambda^t \beta m^t + \sigma_{is} - \nu_c c_p m^t  + \sigma^t_{snu} \perp T^t_{s,n} - T^l_{s,n} \geq 0 \\
0 \leq \sigma^t_{snu} \perp T^u_{s,n} - T^t_{s,n} \geq 0 \\
0 \leq \nu_h c_p m^t - \left(\sigma_{is} - \nu_c c_p m^t\right) \perp \hat{T}^t_i - d^t \hat{T}^t_0 - \left(1-d^t\right) \hat{T}^t_n \geq 0 \\
0 \leq \nu_n c_p m^t - \left(\sigma_{is} - \nu_c c_p m^t\right) \perp T^t_{s,n} - T^t_s \geq 0 \\
0 \leq \sigma^t_{is} \perp \hat{T}^t_i - T^t_s \geq 0 \\
- \hat{T}^t_n + \left(1-\gamma\right) \hat{T}^{t-1}_n + \beta m^t \left(T^t_{s,n} - \hat{T}^t_n \right) + \gamma \hat{T}^t_0 + Q^t_n = 0 \\
\hat{T}^T_n - T^0_n = 0
\end{gather}

\end{document}